%% file: main.tex
\documentclass[9.5pt,journal,compsoc]{IEEEtran}
\usepackage{xspace}
\usepackage{soul}
\usepackage{url}
\usepackage{graphicx}
\usepackage{amsmath}
\usepackage{amsfonts}
\usepackage{multirow}
\usepackage{amsthm}
\usepackage{booktabs}
\usepackage{epstopdf}
\usepackage[table]{xcolor}
\usepackage[labelformat=simple]{subcaption}
\usepackage{caption}
\usepackage{bbm}
\usepackage{dsfont}
\usepackage[ruled,linesnumbered]{algorithm2e}
\usepackage{setspace}
\usepackage{braket}
\usepackage[normalem]{ulem}
\useunder{\uline}{\ul}{}
\theoremstyle{definition}
\newtheorem{theorem}{Theorem}[section]
\newtheorem{corollary}[theorem]{Corollary}

\newtheorem{myDef}{Definition}[section]
\newtheorem{example}{Example}
\newtheorem{problem}{Problem}
\newcommand{\tabincell}[2]{\begin{tabular}{@{}#1@{}}#2\end{tabular}}
\newcommand{\ourmethod}{GSim\xspace}

\ifCLASSOPTIONcompsoc
  \usepackage[nocompress]{cite}
\else
  \usepackage{cite}
\fi

\ifCLASSINFOpdf
\else
\fi

\begin{document}

\title{\ourmethod: A Graph Neural Network based Relevance Measure for Heterogeneous Graphs}


%
%
%

\author{Linhao~Luo,~Yixiang~Fang,~Moli~Lu,~Xin~Cao,~Xiaofeng~Zhang,~Wenjie~Zhang
\IEEEcompsocitemizethanks{\IEEEcompsocthanksitem Linhao Luo, Moli Lu and Xiaofeng Zhang are with the Department of Computer Science, Harbin Institute of Technology, Shenzhen (E-mail: \{luolinhao,21S051052\}@stu.hit.edu.cn; zhangxiaofeng@hit.edu.cn.).}
\IEEEcompsocitemizethanks{\IEEEcompsocthanksitem Yixiang Fang is with the School of Data Science, The Chinese University of Hong Kong, Shenzhen (E-mail: fangyixiang@cuhk.edu.cn).}
\IEEEcompsocitemizethanks{\IEEEcompsocthanksitem Xin Cao and Wenjie Zhang are with the University of New South Wales, Australia (E-mail: \{xin.cao,wenjie.zhang\}@unsw.edu.au).}
}

%
%

\markboth{Journal of \LaTeX\ Class Files,~Vol.~14, No.~8, August~2015}%
{Shell \MakeLowercase{\textit{et al.}}: Bare Demo of IEEEtran.cls for Computer Society Journals}
%



\IEEEtitleabstractindextext{%

\input{sections/abstract}

}

\maketitle

\IEEEdisplaynontitleabstractindextext

%
\IEEEpeerreviewmaketitle

\input{sections/introduction}
\input{sections/relatedworks}
\input{sections/preliminary}
\input{sections/approach}
\input{sections/experiment}

\input{sections/conclusion}



\ifCLASSOPTIONcaptionsoff
  \newpage
\fi



\bibliographystyle{IEEEtran}
\bibliography{IEEEabrv,sections/main.bib}

\noindent\textbf{Linhao Luo} received the bachelor degree from the Harbin Institute of Technology, Shenzhen in 2021. He is now a PhD student in the faculty of information and technology at Monash University. His research interests include machine learning, data mining, and graph neural networks.

\noindent\textbf{Yixiang Fang} received the Ph.D. degree from the University of Hong Kong (HKU) in 2017. Currently, he is an associate professor in the School of Data Science at the Chinese University of Hong Kong, Shenzhen. His research interests mainly focus on the areas of data management, data mining, and artificial intelligence over big data.

\noindent\textbf{Moli Lu} is a master student in the Harbin Institute of Technology, Shenzhen. Her research interests focus on graph mining.

\noindent\textbf{Xin Cao} is a senior lecturer in the School of Computer Science and Engineering at the University of New South Wales. His research interests include database management, data mining, big data analytics, and artificial intelligence.

\noindent\textbf{Xiaofeng Zhang} is currently an associate professor with the Department of Computer Science, Harbin Institute of Technology, Shenzhen. His research interests include data mining, machine learning, and graph mining.

\noindent\textbf{Wenjie Zhang} is a Professor in the School of Computer Science and Engineering at the University of New South Wales. Her research interests lie in data management and analytics for large-scale data, especially graph, spatial-temporal, and image data. She serves as an Associate Editor for TKDE, Associate Editor for PVLDB 2022, and (senior) PC member for leading conferences in database and data
mining.

\end{document}

%% file: sections/abstract.tex
\begin{abstract}
Heterogeneous graphs, which contain nodes and edges of multiple types, are prevalent in various domains, including bibliographic networks, social media, and knowledge graphs.
As a fundamental task in analyzing heterogeneous graphs, relevance measure aims to calculate the relevance between two objects of different types, which has been used in many applications such as web search, recommendation, and community detection.
Most of existing relevance measures focus on homogeneous networks where objects are of the same type, and a few measures are developed for heterogeneous graphs, but they often need the pre-defined meta-path.
Defining meaningful meta-paths requires much domain knowledge, which largely limits their applications, especially on schema-rich heterogeneous graphs like knowledge graphs.
Recently, the Graph Neural Network (GNN) has been widely applied in many graph mining tasks, but it has not been applied for measuring relevance yet.
To address the aforementioned problems, we propose a novel GNN-based relevance measure, namely \ourmethod.
Specifically, we first theoretically analyze and show that GNN is effective for measuring the relevance of nodes in the graph.
We then propose a context path-based graph neural network (CP-GNN) to automatically leverage the semantics in heterogeneous graphs.
Moreover, we exploit CP-GNN to support relevance measures between two objects of any type.
Extensive experiments demonstrate that \ourmethod outperforms existing measures.
\end{abstract}

%% file: sections/introduction.tex
\IEEEraisesectionheading
{\section{Introduction}
\label{sec:introduction}}

Nowadays, many real-world data are often modeled as heterogeneous graphs, which contain multiple typed nodes and edges.
As a fundamental topic in network science, similarity measure has been studied for decades and found in various real-world applications, such as web search \cite{jeh2003scaling}, recommendation \cite{pera2013group}, and community detection \cite{zarandi2018community}. 
Conventional studies on similarity measures (e.g., SimRank \cite{jeh2002simrank} and P-rank \cite{zhao2009p}) mainly focus on homogeneous graphs, where objects are of the same type.
However, due to the heterogeneous types of nodes and edges in heterogeneous graphs, it is meaningless to measure the ``similarity'' between two nodes of different types by directly applying these existing similarity measures.
As a result, it is necessary to develop novel measures for quantifying the relevance between two nodes in heterogeneous graphs.
For example, Fig. \ref{fig:HIN} (a) depicts a bibliographic network with four types of nodes (i.e., {\it paper}, {\it author}, {\it venue}, and {\it subject}) and five relations (edge types) among them. As shown in Fig. \ref{fig:HIN} (b), we may want to find an author that is the most relevant to a certain venue.
Unlike the similarity measures in homogeneous graphs which focus on objects of the same type, the relevance measure in heterogeneous graphs should measure the relevance between two objects of different types \cite{shi2014hetesim}.

\begin{figure}[]
    \centering
    \includegraphics[trim=0 0.4cm 0cm 0,clip,width=1\columnwidth]{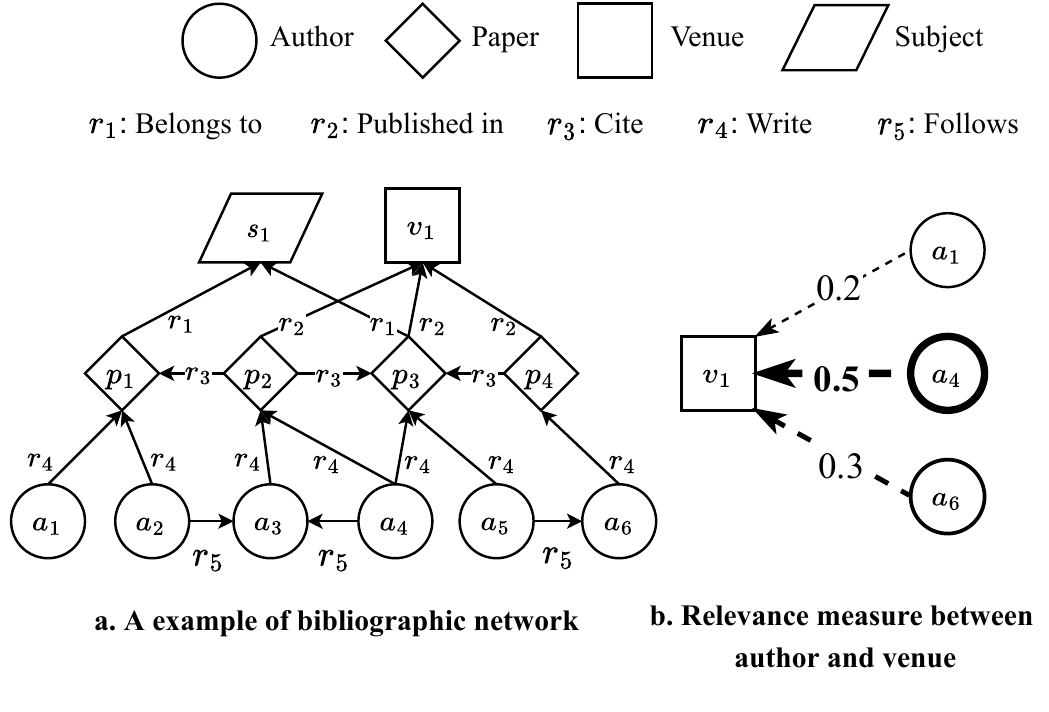}
    \caption{An example bibliographic information network and relevance measure.}
    \label{fig:HIN}
\end{figure}

Compared to measuring the similarity in homogeneous graphs, measuring the relevance in heterogeneous graphs is more challenging, because the multiple types of nodes and edges carry abundant semantic information.
This kind of semantic information plays an important role in measuring the relevance, but it is ignored by many similarity measures in homogeneous networks \cite{jeh2002simrank,jeh2003scaling}.
For instance, in Fig. \ref{fig:HIN} (a), there is no direct connection between authors and venues in the graph, making it hard to measure their relevance directly.
However, if considering the path {\it Author-Paper-Venue}, we can claim that an author is relevant to a venue, as long as she/he has written some papers published in this venue. The path that reveals the semantic information above in the heterogeneous graph is called the {\it meta-path} \cite{sun2011pathsim}.

In recent years, several works have employed the meta-path for measuring the relevance of nodes in heterogeneous graphs, such as
PathSim \cite{sun2011pathsim},
HeteSim \cite{shi2014hetesim},
RelSim \cite{wang2016relsim}, and 
AvgSim \cite{xiao2016avgsim}). 
Nevertheless, a major limitation of these measures is that their performance highly depends on the quality of the pre-defined meta-paths, which needs to be selected manually by domain experts.
Moreover, different meta-paths reveal different semantics, and the number of possible meta-paths increases exponentially with the path length, meaning that it is almost infeasible to find all meta-paths to comprehensively capture the semantics.
Furthermore, they fail to incorporate several meta-paths at the same time to measure the relevance in a collective manner.
In addition, different meta-paths contribute differently to the relevance measure, which imposes great challenges for distinguishing their importance.
Therefore, it is desirable to develop new relevance measures without using pre-defined meta-paths.

With the development of graph representation learning \cite{grover2016node2vec}, the node embeddings that encode the structure and semantics information have been shown effectiveness for measuring the relevance between nodes.
Esim \cite{shang2016esim} is a such kind of measure, but it still requires pre-defined meta-paths.
Recently, Graph Neural Network (GNN) has shown great potential in graph representation learning area \cite{kipf2017semisupervised,hamilton2017inductive,wang2019heterogeneous,yang2021interpretable}.
GNN can not only embed the graph structure information \cite{xu2018powerful,you2021identity,song2022bi} into the node embeddings, but also consider the semantic information \cite{hu2020heterogeneous} in embedding process.
These embeddings can be used for various downstream tasks  \cite{liu2021anomaly,jiang2021could,luo2020motif,yang2022semantically}. An earlier version of this work \cite{luo2021detecting} adopts the GNN to capture the semantic information and uses the learned embeddings for community detection in heterogeneous graphs.
Despite the success in many applications, GNN has not been applied to address the relevance measure problem in heterogeneous graphs. 

Motivated by the above, in this paper we propose a novel GNN-based relevance measure for heterogeneous graphs, which is called \ourmethod.
\ourmethod can automatically leverage the semantics in heterogeneous graphs without using pre-defined meta-paths and learn node embeddings used for the relevance measure.
Specifically, we first theoretically prove that GNN can simulate the meeting probability of \textit{pair-wise random walk}, which has been followed by many previous measure methods.
Based on the theoretical analysis, we propose a context path-based graph neural network (CP-GNN), which adopts the \textit{context path} to capture the semantics between nodes automatically.
Context path \cite{barman2019k} links two nodes with the same type via a sequence of nodes of auxiliary type. It can not only well capture the semantics in heterogeneous graphs, but also avoid selecting meta-paths manually by domain experts.

Besides, we introduce the \textit{relation attention} module in \ourmethod, which calculates the attention score of each relation, so the contributions of different relations are well differentiated. In other words, \ourmethod not only avoids using meta-paths, but also well captures the semantics between nodes preserved by context paths with different importance.
Moreover, to support relevance measures between two nodes of any type, we present the \textit{relation message passing} mechanism that extends the context path to measure the relevance under asymmetric context paths, whose effectiveness has been shown by both theoretical analyses and empirical study.
In addition, we design a \textit{type-length attention} module to capture the fine-grained importance of context paths of different lengths. Last, we adopt the framework of supervised contrast learning to optimize \ourmethod, so that the relative relevance between nodes can be well-preserved.

In summary, our principal contributions are as follow:
\begin{itemize}
    \item We propose a novel GNN-based relevance measure method, called \ourmethod, for heterogeneous graphs. To our best knowledge, this is the first relevance measure for heterogeneous graphs that exploits GNN.
    
    \item We theoretically show that GNN can be used to measure the relevance effectively, based on which a context path-based GNN (CP-GNN) is designed for capturing the semantics of relevance measure.
    
    \item We conduct extensive experiments on four real-world heterogeneous graphs, and the results on tasks like relevance search and community detection demonstrate the superior performance of \ourmethod.
\end{itemize}

An earlier version of this work was published in the conference CIKM'2021 \cite{luo2021detecting}, where we proposed the CP-GNN for community detection in heterogeneous graphs. In this paper, we extend the CP-GNN to measure the relevance in heterogeneous graphs and provide theoretical analysis as well as extensive experiments to guarantee effectiveness.
The paper is organized as follows. We review the related works in Section \ref{sec:relatedworks}. We describe the related notations and problem definition in Section \ref{sec:pre}. Section \ref{sec:methodology} discusses our proposed \ourmethod. We report the experimental results in Section \ref{sec:experiments} and conclude in Section \ref{sec:conclusion}.

%% file: sections/relatedworks.tex
\section{Related work}
\label{sec:relatedworks}

\subsection{Similarity and Relevance Measures}

Similarity measure, aiming to calculate the similarity between objects of the same type, has been studied for decades \cite{zhao2009p,rothe2014cosimrank,li2022one}. Previous methods often utilize the structure to measure the similarity. 
For example, SimRank \cite{jeh2002simrank} calculates the similarity of two nodes by recursively averaging the similarity of their neighbors.
Personalized PageRank (PPR) \cite{jeh2003scaling} measures the similarity by using the probability of random walks starting from the source node to the target node. 
However, since these approaches only consider objects of the same type, they cannot be directly used in heterogeneous graphs.
    
To extend the application to heterogeneous graphs, a few relevance measures have been developed.
For example, PathSim\cite{sun2011pathsim} introduces the concept of meta-path-based similarity by calculating possible meta-path instances, but it can only use symmetric meta-paths to measure the similarity between nodes of the same type.
To measure the relevance of different-type objects, HeteSim \cite{shi2014hetesim} computes the meeting probability of two nodes following a given meta-path.
RelSim \cite{wang2016relsim} further introduces the latent semantic relation (LSR) by combining weighted meta-paths to catch different semantics for relevance measure.
SCHAIN\cite{li2017schain} and CMOC-AHIN\cite{zhou2020cmocahin} calculate the relevance as a weighted sum of all attributes and meta-paths.
However, all the aforementioned methods need user-defined meta-paths. Besides, HeteSim and PathSim can use only one meta-path at a time, which does not fully excavate the semantics of different meta-paths. In addition, when multiple meta-paths are used, how to properly assign their weights is also a challenging task. 

Recently, some meta-path-free measures have been studied. nSimGram\cite{conte2018nsimgram} makes use of q-gram and applies the Bray-Curits similarity index to measure the relevance, but it can only be applied on the labeled heterogeneous graphs.
HowSim\cite{wang2020effective} is another meta-path-free method that aggregates similarities over relations to catch different semantics, but it needs the two query nodes to be of the same type, thus cannot handle our relevance measure problem.

\subsection{Graph Neural Network}
Conventional graph representation learning (e.g., DeepWalk \cite{perozzi2014deepwalk}, LINE \cite{tang2015line}, Node2Vec \cite{grover2016node2vec}, and Metapath2vec \cite{dong2017metapath2vec}) aims at mapping node into a low-dimensional vector, which preserves the original graph information and can be used for various tasks.
For instance, ESim \cite{shang2016esim} learns node embeddings under the guidance of meta-paths and uses the embeddings for relevance measure in heterogeneous graphs.
    
Recently, Graph Neural Network (GNN) has been shown powerful for learning node representation \cite{xu2022uncertainty}. The key idea of GNN is to aggregate information from node’s neighbors via neural networks \cite{gnn_tnn09}.
For example, GAT \cite{velivckovic2017graph} applies the attention mechanism on each node and its neighbors, which can help the center node aggregate information from the important neighbors. In this way, GNN can embed both the structure and feature information into the node embedding as discussed in previous works \cite{xu2018powerful,you2021identity}.
GNN has also been studied on heterogeneous graphs. For example, MEIRec \cite{fan2019metapath} is a meta-path guided heterogeneous GNN that learns the embeddings of objects for intent recommendation. HAN \cite{wang2019heterogeneous} adopts the meta-path to leverage the semantic information and uses the attention mechanism to differentiate them. MAGNN \cite{fu2020magnn} further aggregates information along the meta-path to incorporate fine-grained semantics. ie-HGCN \cite{yang2021interpretable} designs a hierarchical aggregation
architecture to automatically extract useful meta-paths, which presents good interpretability and model efficacy.
However, all these methods still require the meta-path, which largely limits their applications. 
    
Although some recent GNN-based models (e.g., HGT \cite{hu2020heterogeneous}) have considered the semantics in heterogeneous graphs, none of them is designed for relevance measurement.
On the other hand, the GNN-based embedding models have been shown more powerful than traditional similarity metric-based embedding models, which can be designed by applying cosine similarity \cite{singhal2001cosine}, Jaccard coefficient \cite{levandowsky1971jaccard}, and the p-norm distance \cite{duren1970pnorm} on them.

Motivated by the above, in this paper we propose a GNN-based relevance measure, namely \ourmethod, which is not only meta-path-free but also universal that can be used on different-typed objects, allowing us to measure the node relevance in heterogeneous graphs in an end-to-end manner.

%% file: sections/preliminary.tex
\section{Preliminary}
\label{sec:pre}

In this section, we first briefly introduce the data model of the heterogeneous graph, and then present some key concepts and the problem that we study in this paper.

\begin{myDef}[\textbf{Heterogeneous graph}]
The heterogeneous graph is defined as a graph $\mathcal H =(\mathcal{V},\mathcal{E},\mathcal{A}, \mathcal{R})$ with a node mapping function $\phi(v): \mathcal{V}\to \mathcal{A}$ and an edge mapping function $\psi(e):\mathcal{E}\to \mathcal{R}$,
where $|\mathcal{A}| + |\mathcal{R}| > 2$, each node $v\in\mathcal V$ belongs to a node type $\phi(v)\in\mathcal A$, and each edge $e\in\mathcal E$ belongs to an edge type (also called relation) $\psi(e)\in\mathcal R$.
\end{myDef}

\begin{myDef}[\textbf{Meta-path \cite{sun2011pathsim}}]
Given a heterogeneous graph $\mathcal H$, the meta-path is a path with the form $\mathcal{P} = A_0 \xrightarrow{R_1} A_1 \xrightarrow{R_2}\cdots \xrightarrow{R_{k}} A_{k}$, where $A_i\in\mathcal A$, $R_i\in\mathcal R$ ($1\leq i\leq k$), and $R_1 \circ R_2 \circ \cdots \circ R_{k}$ defines the composite relation between node type $A_0$ to node type $A_{k}$. The length of the meta-path is the number of composite relations $|R_1 \circ R_2 \circ \cdots \circ R_{k}|$.
\end{myDef}

\begin{myDef}[\textbf{Context path \cite{barman2019k}}]\label{def:contextpath}
Given a heterogeneous graph $\mathcal H$, a {\it context path} is a path connecting two nodes $v_i$ and $v_j$ with the same type $A$, formulated as $\rho^k=\{ v_i, R^{k-1}, v_j \}$. The $R^{k-1}$ is any path connecting $v_i$ and $v_j$ that contains $k-1$ ($k\geq1$) nodes of \textit{auxiliary types} $\mathcal{A}'=\mathcal{A}\backslash\{A\}$, where nodes in $R^{k-1}$ is also denoted as the \textit{auxiliary nodes}. The length of a context path is $k$ (when $k=1$, $R^k=\emptyset$).
\end{myDef}

\begin{figure}[]
    \centering
        \includegraphics[trim={0 0 0 4},width=0.7\columnwidth,clip]{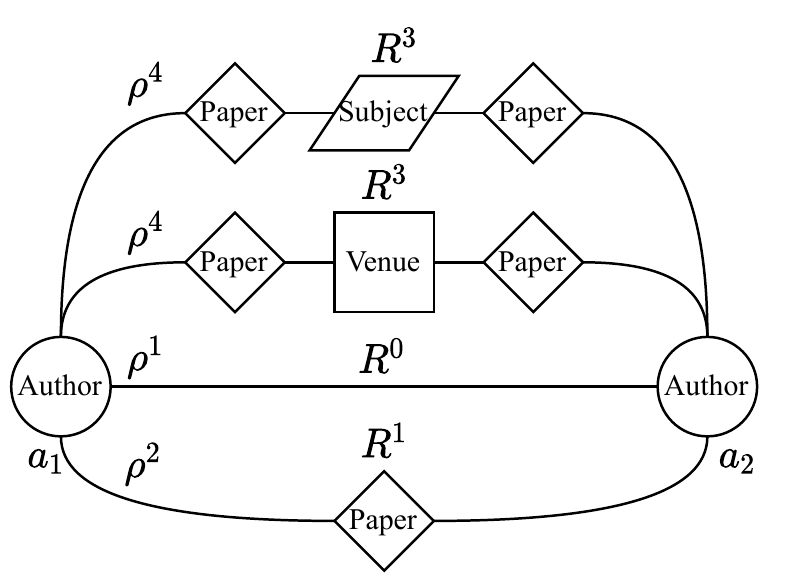}
    \caption{An example of four context paths between $a_1$ and $a_2$.}
    \label{fig:ContextPath}
\end{figure}

\begin{example}
Fig. \ref{fig:ContextPath} depicts four possible context paths $\rho^*$ of different lengths that connect authors $a_1$ and $a_2$, where $R^*$ denotes the auxiliary nodes that constitute the path.
\end{example}

\noindent\textbf{Difference between context path and meta-path.} Intuitively, different semantic relationships revealed from the context path come from different auxiliary nodes. The purpose of the meta-path is to manually define the combination of auxiliary nodes in the context path. However, the number of combinations explodes exponentially as the node types and path length increase. Thus, the context path relaxes the restriction of the auxiliary nodes. Given a length $k$, the context path contains all the possible $k$-order relationships. For example, in Fig. \ref{fig:ContextPath}, two meta-paths \textit{Author-Paper-Subject-Paper-Author} and \textit{Author-Paper-Venue-Paper-Author} can both be represented by a 4-length context path.

Besides, specifying an integer value of $k$ is much easier than defining the meta-path. This is because the number of meta-paths of different node/edge types grows exponentially as the meta-path length increases, while the number of possible values of $k$ is rather limited since the average length of the shortest paths between any two nodes in real-world networks is often between 4 to 6, according to \cite{ye2010distance}.

Noticeably, the original definition of context path only considers the path between nodes of the same type, which cannot measure nodes of heterogeneous types. Therefore, we propose the \textit{generalized context path} to extend the CP-GNN model in Section \ref{sec:cp-gnn+}, which is defined as follows:
\begin{myDef}[Generalized context path]\label{def:gcontextpath}
    Generalized context path relaxes the same type constraint for nodes $v_i,v_j$ in the contex path $\rho^k=\{v_i, R^k, v_j\}$. Thus, it can capture the semantics between nodes of different types.
\end{myDef}

\begin{problem}
Given a heterogeneous graph $\mathcal H$ and two nodes $v_i$ and $v_j$ with any type in $\mathcal H$, we want to design a measure $S(v_i, v_j)$ such that it can accurately capture the relevance between $v_i$ and $v_j$ in $\mathcal H$.
\end{problem}

%% file: sections/approach.tex
\section{Our Approach}
\label{sec:methodology}

In this section, we first discuss how GNN can measure the relevance on graph. Then, we propose a context path-based graph neural network (CP-GNN) to automatically leverage the semantics in heterogeneous graphs. Finally, we propose the CP-GNN+ to support relevance measures between any type of object.

\subsection{GNN for relevance measure}\label{sec:theoreticalanalysis} 
Previous measure methods (e.g., SimRank \cite{jeh2002simrank} and HeteSim \cite{shi2014hetesim}) are based on the theory of \textit{pair-wise random walk}.
\begin{myDef}[\textbf{Pair-wise random walk (PRW) \cite{jeh2002simrank}}]\label{def:prw}
    Pair-wise random walk $\text{PRW}(v_i, v_j)$ measures the probability that two random walks starting from $v_i$ and $v_j$ meet at the same node on the graph. This can be formulated as 
    \begin{equation}
        \label{eq:prw}
        PRW(v_i, v_j|\pi^k) = \sum_{\pi^k}p(v_i,v_j|\pi^k),
    \end{equation}
    where $\pi^k: \{v_i,\cdots,v_m,\cdots v_j\}$ denotes arbitrary $k$-length walks (path contains $k$ edges) that $v_i,v_j$ meet at intermediate node $v_m$. The higher the probability, the more likely the two nodes are similar.
    
    The $\pi^k$ can be seen as the composition of two walks $\pi_{i:m}, \pi_{j:m}$ respectively starting from $v_i$ and $v_j$ and ending at $v_m$. Thus, the $p(v_i,v_j|\pi^k)$ can be written as
    \begin{equation}
        \label{eq:prw_pro}
        p(v_i,v_j|\pi^k)=p(v_m|v_i, \pi_{i:m})p(v_m|v_j, \pi_{j:m})
    \end{equation}
    where $p(v_m|v_*, \pi_{*:m})$ denotes the probability of $v_*$ visits at node $v_m$ under path $\pi_{*:m}$. Following the idea of random walk, each node will randomly visit one of its neighbors, we formulate $p(v_m|v_*, \pi_{*:m})$ as
    \begin{equation}
        p(v_m|v_*, \pi_{*:m}) = \prod_{w_l\in\pi_{*:m}} \frac{1}{O(w_l)},
    \end{equation}
    where $O(w_l)$ denotes the out degree of $l$-th node in $\pi_{*:m}$.
\end{myDef}

HeteSim, to capture the semantics in heterogeneous graphs, uses the meta-path $\mathcal{P}$ as constraint, which measures how likely $v_i$ and $v_j$ will meet at the same node when they travel along the meta-path. Given a meta-path $\mathcal{P}= A_0 \xrightarrow{R_1} A_1 \xrightarrow{R_2}\cdots \xrightarrow{R_{k}} A_{k}$, it can be formulated as
\begin{equation}
    \label{eq:hetesim}
    \text{HeteSim}(v_i, v_j|\mathcal{P}) = \sum_{\varphi^k \in \mathcal{P}}p(v_i,v_j|\varphi^k),
\end{equation}
where $\varphi^k$ denotes the $k$-length meta-path instance constrained by $\mathcal{P}$. Equation \ref{eq:hetesim} shows that HeteSim needs to iterate over all possible meta-path instances and sum up the relevance.

Extending the theory of pair-wise random walk, we propose the theorem that GNN can simulate pair-wise random walk and measure the relevance between two nodes $v_i$ and $v_j$. 
\begin{theorem}\label{thm:gnn_prw}
    The inner-product of two node representations $h_i^k,h_j^k$ generated by a $k$-layer GNN is equal to the probability of pair-wise random walk under $2k$-length paths. This can be formulated as
    \begin{align}
        &H^k = \text{k-layer GNN}(Z),\\
        &PRW(v_i, v_j|\pi ^{2k}) = \Braket{h_i^k, h_j^k},
    \end{align}
    where $Z$ denotes the initial node embedding, $h_i^k, h_j^k\in H^k$, and $\Braket{\cdot, \cdot}$ denotes the inner-product operation.
\end{theorem}


\begin{proof}
    Following the Equation \ref{eq:prw_pro}, a $2k$-length path $\pi^{2k}$ can be divided into two $k$-length walks, formulated as
    \begin{equation}
        \pi^{2k}=\{\pi_{i:m}^k,\pi^k_{m:j}\}.
    \end{equation}

    Thus, we can prove that each entry of the node representation $h_i^k[l]$ denotes the probability that node $v_i$ visits $v_l\in\mathcal{V}$ under $k$-length walks, written as $p(v_l|v_i, \pi^{k})$.

    A $k$-layer GNN can be written as
    \begin{equation}
        H^k = \sigma \Big(L\cdots\sigma\big(L\sigma(LZW^1)W^2\big)\cdots W^k\Big),
    \end{equation}
    where $Z$ denotes the initialized node embedding, $L$ denotes the Laplacian matrix, $\sigma(\cdot)$ denotes the activation function, and $W^*$ denotes the weight matrix.
    
    For simplicity, we can assign the $Z$ with the unique \textit{one-hot embedding}, where $Z\in\mathbb{R}^{|V|\times|V|}=\textbf{I}$, and use $D^{-1}A$ as the Laplacian matrix, where $D$ denotes the degree matrix, and $A$ denotes the adjacency matrix. The $D^{-1}A$ is also known as the random walk transition matrix, which means each node will walk to its neighbors with equal chance. The $\sigma(\cdot)$ and $W^*$ are also defined as the identity matrix $\textbf{I}$.

    When $k=1$, the $H^1$ can be formulated as
    \begin{equation}
        H^1 = \textbf{I}D^{-1}AZ\textbf{I} = D^{-1}AZ = D^{-1}A,
    \end{equation}
    where $H^1\in\mathbb{R}^{|V|\times|V|}$.
    Given a node $v_i$, each entry of its node representation can be written as
    \begin{equation}
        h^1_i[l]=\left\{ 
        \begin{aligned}
            &\frac{1}{O(v_i)}, &A[i][l]=1,\\
            &0, &else.
        \end{aligned}
        \right.
    \end{equation}
    Therefore, $h^1_i[l]$ denotes the probability that node $v_i$ directly transits to $v_l$, written as $p(v_l|v_i,\pi^1)$.

    \begin{figure*}[h]
        \centering
        \includegraphics[width=.9\textwidth]{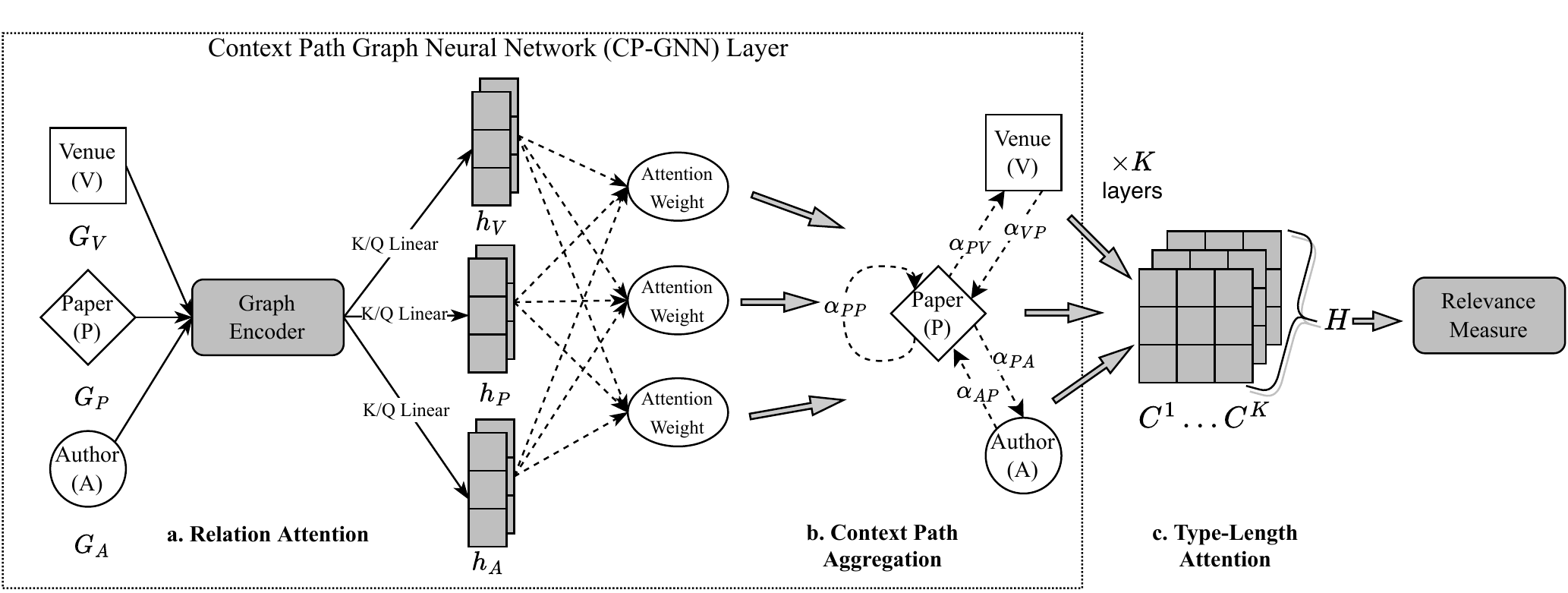}
        \caption{The overall framework of the CP-GNN. a. Relation attention assesses the importance of different relations by calculating the attention score. b. Context path aggregation aggregates the information to generate the context information vector. c. Type-length attention summarizes the vectors under different lengths and types to obtain the final node embeddings.}
        \label{fig:framework}
    \end{figure*}

    Assuming at layer $k-1$, $h^{k-1}_i[l]$ denotes the probability that node $v_i$ visits $v_l$ under $(k-1)$-length path, written as $p(v_l|v_i, \pi^{k-1})$. At $k$ layer, the $h^{k}_i[l]$ is formulated as
    \begin{equation}
        \label{eq:gnn_prw}
        \begin{aligned}
            h^{k}_i[l] &= \sum_{v'_l\in N(v_l)}\frac{1}{O(v'_l)} h^{k-1}_i[l']\\
            &= \sum_{v'_l\in N(v_l)} p(v_l|v'_l,\pi^1)p(v'_l|v_i, \pi^{k-1})\\
            &= p(v_l|v_i, \pi^{k}),
        \end{aligned}
    \end{equation}
    where $N(v_l)$ denotes the direct neighbors of $v_l$. From Equation \ref{eq:gnn_prw}, we can see that GNN summarizes the transition probability from neighbors to synthesize visit probability.

    Therefore, the inner production of $h^{k}_i,h^{k}_j$ can be formulated as 
    \begin{equation}
        \setlength\abovedisplayskip{2pt}
        \setlength\belowdisplayskip{2pt}
        \begin{aligned}
            \label{eq:prw_prove}
            \Braket{h_i^k, h_j^k} &= \sum_{l=1}^{|\mathcal{V|}}h^{k}_i[l] h^{k}_j[l]\\
            &= \sum_{l=1}^{|\mathcal{V|}}p(v_l|v_i, \pi^{k})p(v_l|v_j, \pi^{k})\\
            &= \sum_{\pi^{2k}}p(v_i,v_j|\pi^{2k})\\
            &= PRW(v_i, v_j|\pi^{2k}).
        \end{aligned}
    \end{equation}
    From Equation \ref{eq:prw_prove}, we can see that $\Braket{h_i^k, h_j^k}$ summarizes the meeting possibility of every node in the graph, which is the same as the Definition \ref{def:prw}.
\end{proof}

From Theorem \ref{thm:gnn_prw}, we can see that by stacking $k$ layers, GNN can simulate the $k$-length random walk and inject the reaching possibility into the node embedding. Thus, the inner-product of any two embeddings summarizes the meeting possibility under $2k$-length paths, which is equal to the definition of the pair-wise random walk. 
Based on Theorem \ref{thm:gnn_prw}, we can define a GNN-based relevance measure method, which is formulated as
\begin{align}
    \setlength\abovedisplayskip{2pt}
    \setlength\belowdisplayskip{2pt}
    &H^k = \text{k-layer GNN}(Z),\\
    &S(v_i, v_j) = \Braket{h_i^k, h_j^k}, h_i^k, h_j^k\in H^k.
\end{align}

\subsection{Context path-based graph neural network}

To extend the Theorem \ref{thm:gnn_prw} for heterogeneous graphs, we propose a novel Context Path-based Graph Neural Network (CP-GNN), which aims to learn node representations that are able to well capture the semantic information for relevance measure.
The overall framework of CP-GNN is depicted in Fig. \ref{fig:framework}.


As we discussed in section \ref{sec:pre}, the context path can well represent the semantic information in heterogeneous graphs. Therefore, we propose the CP-GNN that recursively embeds the semantics of each node into a context information vector $c^k$, which can be formulated as
\begin{equation}
    C^k = \text{k-layer CP-GNN}(Z),
\end{equation}
where $Z\in \mathbb{R}^{|V|\times d}$ denotes the initialized node embeddings, and $c^k \in C^k$ denotes each nodes's $k$-length context information vector. 

Instead of using the random walk transition matrix, CP-GNN consists of two major components (i.e., {\it relation attention} and {\it context path aggregation}) to control the transition probability. In this way, we can automatically differentiate the importance of different semantics in the context paths. Then, we propose a {\it type-length attention} mechanism to synthesize information of different length context paths for relevance measure.

\subsubsection{Relation Attention} 
Relation attention aims to calculate the attention score of each relation, so that the contributions of different relations are well differentiated.
We first use a graph encoder to encode the graph of each node type into a summary vector. After that, the attention score of each relation is calculated based on the graph summary vectors.

Given a node type $A$, the $G_A$ is denoted by $G_A$=$(\mathcal{V}_A, \mathcal{E}_A)$ where each node $v\in \mathcal{V}_A$ is of the node type $A$ and each edge $e \in \mathcal{E}_A \subseteq \{\mathcal{V}_A \times \mathcal{V}_A\}$. To enhance the model robustness, the graph encoder contains a node dropout mechanism that randomly drops nodes from the original graph. Then an averaging operation is adopted to calculate the graph summary vector $h_A$.
Although there exist several techniques to generate the graph summary vector $h_A$, the simple averaging operation demonstrates superior performance \cite{ren2019heterogeneous}, and thus $h_A$ is calculated as
    \begin{align}
        \setlength\abovedisplayskip{2pt}
        \setlength\belowdisplayskip{2pt}
        G_A' &= NodeDropout(G_A),\\
        h_A  &= Mean(C_A'),
    \end{align}
where $C_A'$ denotes context information vectors of all the nodes in the remaining graph $G_A'$.

After calculating the $h_A$, for the $l$-th CP-GNN layer, we calculate the $h$-head attention score $\alpha_{S,T}^{h,l}$ for each relation $r_{S,T} \in \mathcal{R}$ by
    \begin{align}
        \setlength\abovedisplayskip{2pt}
        \setlength\belowdisplayskip{2pt}
        & h_S^l = GraphEncoder(C_S^{l-1}),\\
        & h_T^l = GraphEncoder(C_T^{l-1}),\\
        & \alpha_{S,T}^{h,l} = \mathop{Softmax}\limits_{S\in \mathcal{A}} \frac{Q^h(h_T^l)^\top K^h(h_S^l)}{\sqrt{d}},\\
        & Q^h(h_T^l) = QLinear_{T}^h(h_T^l),\\
        & K^h(h_S^l) = KLinear_{S}^h(h_S^l),
    \end{align}
where $S$ and $T$ respectively denote the source and target node types in the relation $r_{S,T}$, $h_S^l$ and $h_T^l$ denote the graph summary vector of $G_S$ and $G_T$ at layer $l$ respectively, $C_S^{l-1}$ and $C_T^{l-1}$ are the context information vectors at the ($l-1$)-th layer, $QLinear$ and $KLinear$ are the linear projection functions that project the graph summary vectors to a Query vector and a Key vector. We want to learn more diverse importance of the relations, thus we adopt total $H$ different heads of relation attention with their own parameters to be learned during the training. $\alpha_{S,T}^{h,l}$ is the attention weight in head $h$ at layer $l$ for the relation $r_{S,T}$.

\subsubsection{Context Path Aggregation} 

Context path aggregation aims to aggregate the information along relations to generate the context information vectors for all nodes. After calculating the scores of different relationships, we aggregate the information for a node $v_i$ of type $T$ from its one-hop neighbors by adopting the widely used GNN message passing method. 

Assuming at layer $l$, we are going to obtain the context information vector $c_i^{l}$ for $v_i$ by aggregating the information from its neighbors along different relations.
We utilize its neighbors' context information vectors obtained at layer $l-1$, which can be formulated as
\begin{equation}\label{eq:civ}
    \resizebox{\columnwidth}{!}{$
    c_i^{l} = W_2^{l}\Big(
        ||_{h\in[1,H]} \sigma(
                W_1^{l} \sum\limits_{r_{S,T} \in \mathcal{R}} \alpha_{S,T}^{h,l} \sum\limits_{v_j \in N_S(i)} c_j^{l-1} + B_1^{l}
            ) + B_2^{l}
        \Big),
    $}
\end{equation}
where $N_S(i)$ denotes the adjacent neighbors of $v_i$ in graph $G_S$ for each relation $r_{S,T} \in \mathcal{R}$ relevant to node type $T$, $W_1^l$, $W_2^l$, $B_1^l$, and $B_2^l$ are the trainable parameters in the $l$-th layer, and $H$ is the number of attention heads. Note that finally we only use the embedding of nodes at the $k$-th layer to obtain $k$-length context information vectors $C^k$.

In order to get the information of $k$-length context paths, we stack $k$ CP-GNN layers to obtain the $C^k$ at layer $k$. The GRU mechanism \cite{cho2014learning} is also utilized to alleviate the over smoothing problem that unusually occurred in GNN model \cite{chen2020measuring}. This can be formulated as
\begin{gather}
    \setlength\abovedisplayskip{2pt}
    \setlength\belowdisplayskip{2pt}
    \hat{C}^l = \text{CP-GNN Layer}(C^{l-1}), l\in [1,k] \\
    C^{l} = \text{GRU}\Big(C^{l-1}, \hat{C}^l\Big).
\end{gather}
Therefore, the final context information vector $c_i^k$ of each node in $\mathcal H$ can be taken from $C^k$.

\subsubsection{Relation Message Passing}\label{sec:cp-gnn+} 
From the Definition \ref{def:contextpath} of context path, we can see that CP-GNN can only measure the relevance between the nodes of the same type. To measure nodes of heterogeneous types, we can use the \textit{generalized context path}\footnote{We still use the term of \textit{context path} to denote \textit{generalized context path} in the following sections for simplicity.} in Definition \ref{def:gcontextpath} to extend CP-GNN. However, the context paths are often asymmetric between nodes of different types. For example, one possible context path between author and paper is {\it Author-Paper-Subject-Paper}. There are not intermediate nodes between \textit{Paper} and \textit{Subject} for walks meeting. Thus, we cannot adopt CP-GNN to capture such semantics. To address this challenge, we propose the \textit{relation message passing} mechanism and integrate it with CP-GNN to come up with the CP-GNN+.

Inspired by HeteSim \cite{shi2014hetesim}, we can simply add an intermediate node $E$ to each edge in the original graph. For example, path {\it Author-Paper-Subject-Paper} can be transformed into {\it Author-$E_{AP}$-Paper-$E_{PS}$-Subject-$E_{SF}$-Paper}. Thus, arbitrary context paths can be transformed into even-length paths, where walks can meet. We can prove that adding intermediate nodes would not change the relevance of nodes in the original graph.
\begin{theorem}
    \label{thm:intermediatenode}
    Adding intermediate nodes would not affect the relevance calculated by GNN. 
\end{theorem}
\begin{proof}
    We can denote the set of added intermediate nodes by $\mathcal{V}_{new}$ and the graph after adding nodes by $\tilde{\mathcal{H}}$. The nodes in $\tilde{\mathcal{H}}$ are denoted by $\tilde{\mathcal{V}}=\mathcal{V}\cup \mathcal{V}_{new}$. We can extend the dimension of initialized node embedding to $|\tilde{\mathcal{V}}|$, and assign one-hot embeddings for $v_{new}$.

    Before adding intermediate nodes, the visit probability under $k$-length paths is written as 
    \begin{equation}
        p(v_l|v_i, \pi^{k}) = \prod_{w_l\in\pi^k} \frac{1}{O(w_l)}.
    \end{equation}
    After adding intermediate nodes, the $k$-length path is extended to $2k$-length.
    Since the out-degree of an intermediate node is 1, the probability under extended $2k$-path is written as
    \begin{align}
        \setlength\abovedisplayskip{2pt}
        \setlength\belowdisplayskip{2pt}
        p(v_l|v_i, \pi^{2k}) &= \prod_{w_l\in\pi^{2k}} \frac{1}{O(w_l)}\\
        &= \prod_{w_l\in\pi^{2k} \wedge w_l \in\mathcal{V} }\frac{1}{O(w_l)}\prod_{w_l\in\pi^{2k} \wedge w_l' \in\mathcal{V}_{new} }\frac{1}{O(w_l')}\\
        &= \prod_{w_l\in\pi^{2k} \wedge w_l \in\mathcal{V} }\frac{1}{O(w_l)}\\
        &=p(v_l|v_i, \pi^{k}).
    \end{align}
    According to Theorem \ref{thm:gnn_prw}, we can simply use a $2k$-layer GNN on $\tilde{\mathcal{H}}$ to obtain the relevance in original graph.
\end{proof}

Though adding intermediate nodes would not affect the relevance, the storage complexity of the graph will increase from $O(|\mathcal{V}|+|\mathcal{E}|)$ to $O(|\mathcal{V}|+2|\mathcal{E}|)$. This is unacceptable for large-scale graphs.

To address the aforementioned challenge, we propose the \textit{relation message passing} mechanism, which would not increase the storage complexity. We introduce the Theorem \ref{thm:relationmessagepassing} to facilitate the analysis. 
\begin{theorem}\label{thm:relationmessagepassing}
    There exists an injective function $f$ that can represent each intermediate node $v_{e_{ij}}$ with its neighbor nodes, which can be formulated as
    \begin{equation}
        m_{e_{ij}} = f(z_i, z_j), v_{e_{ij} \in V_{new}}, v_i, v_j\in N(v_{e_{ij}}),
    \end{equation}
    where $v_{e_{ij}}$ denotes the added intermediate node for original edge $e_{ij}=(v_i, v_j)$, and $z_*$ denotes the node embedding.
\end{theorem}
\begin{proof}
    According to Theorem \ref{thm:intermediatenode}, added nodes with unique representation would not affect the relevance results. In Theorem \ref{thm:gnn_prw}, we assign the unique one-hot embedding for $v \in \mathcal{V}$. Additionally, each $v_{e_{ij}}$ only connects to nodes $v_i$ and $v_j$. Thus, by defining the $f$ as a simple add operation, we can generate at most $\tbinom{|\mathcal{V}|}{2}$ unique embedding vectors for each intermediate node. This can be formulated as 
    \begin{equation}
        \setlength\abovedisplayskip{2pt}
        \setlength\belowdisplayskip{2pt}
        m_{e_{ij}} = z_i + z_j.\label{eq:int_add}
    \end{equation}
\end{proof}

Based on Theorem \ref{thm:intermediatenode} and \ref{thm:relationmessagepassing}, we introduce the \textit{relation message passing} mechanism, which is shown in Fig. \ref{fig:relationmessagepassing}. 
In relation message passing, we first synthesize the representation for intermediate node. Practically, the dimension of node embedding is $d<<|\mathcal{V}|$. The node representation cannot be a one-hot embedding. Thus, simply adding them cannot generate unique embeddings for intermediate nodes.
Besides, due to different types of nodes and edges, a simple adding operation in Equation \ref{eq:int_add} would ignore such semantic information. Based on the universal approximation theorem \cite{hornik1989multilayer}, we adopt the MLP to learn the mapping function $f$, which can be formulated as
\begin{equation}
    \setlength\abovedisplayskip{2pt}
    \setlength\belowdisplayskip{2pt}
    m^{l-1}_{e_{ij}} = W_r(c^{l-1}_i||c^{l-1}_j), v_i,v_j \in N(v_{e_{ij}}),
\end{equation}
where we use the relation-specific parameter $W_r$ to learn mapping functions for different types of relations.

Then, we integrate the representation with context path aggregation to generate the final node embedding. This can be formulated as 
\begin{equation}
    \setlength\abovedisplayskip{2pt}
    \setlength\belowdisplayskip{2pt}
    \resizebox{\columnwidth}{!}{$
    c_i^{l} = W_2^{l}\Big(
        ||_{h\in[1,H]} \sigma(
                W_1^{l} \sum\limits_{r_{S,T} \in \mathcal{R}} \alpha_{S,T}^{h,l} \sum\limits_{v_j \in N_S(i)} m^{l-1}_{e_ij} + B_1^{l}
            ) + B_2^{l}
        \Big).
    $}
\end{equation}
The relation message passing extends the context path and CP-GNN to measure the relevance between nodes of different types. Besides, it would not increase the storage complexity.

By using the relation message passing, we propose the CP-GNN+ for relevance measure in heterogeneous graphs. 
From Theorem \ref{thm:gnn_prw}, \ref{thm:intermediatenode}, and \ref{thm:relationmessagepassing}, we can easily have the following corollary.
\begin{corollary}\label{coro:cp_gnn+}
    A $k$-layer CP-GNN+ can simulate the probability that pair-wise random walk meets at $(k+1)$-length context path.
\end{corollary}
The Corollary \ref{coro:cp_gnn+} theoretically ensures the effectiveness of CP-GNN+.

\begin{figure}[t]
    \centering
    \includegraphics[trim=0 0cm 0cm 0cm,clip,width=0.65\columnwidth]{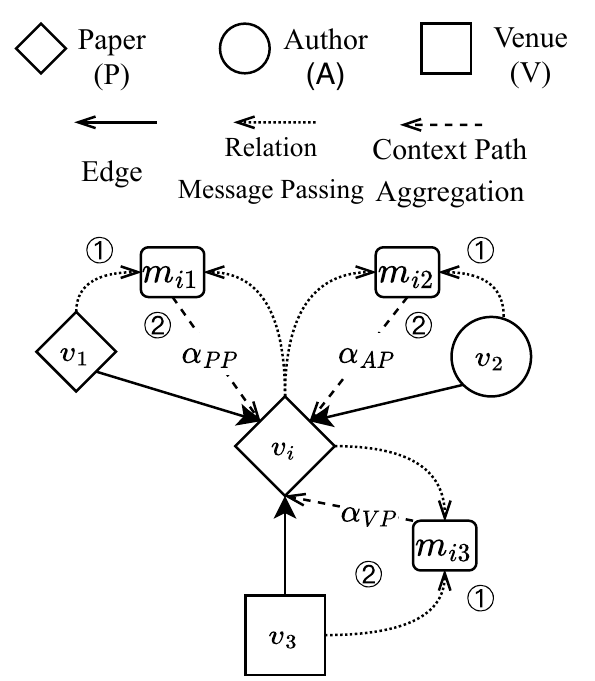}
    \caption{The illustration of relation message passing. It first synthesizes the representation for intermediate nodes, and then integrates the representation with context path aggregation to generate the final node embedding.}
    \label{fig:relationmessagepassing}
\end{figure}

\subsection{\ourmethod for relevance measure}
In this section, we introduce the process of \ourmethod for relevance measure.
Given a heterogeneous graph $\mathcal{H}$, we want to capture the semantics under different lengths. By defining a maximum length $K$, we obtain context information vectors of different lengths, which are formulated as
\begin{equation}
    \setlength\abovedisplayskip{2pt}
    \setlength\belowdisplayskip{2pt}
    C^k = \text{k-layer CP-GNN+}(Z), k \in [1,K].
\end{equation}
Then, we propose a \textit{type-length attention} $\alpha_A^k, A\in \mathcal{A}$ to differentiate the importance of various lengths for each node type, which is formulated as 
\begin{align}
    \setlength\abovedisplayskip{2pt}
    \setlength\belowdisplayskip{2pt}
    H_A&=\sum_{k=1}^K \alpha_A^k C_A^k, A\in \mathcal{A},\\
    H&=\{H_A|A\in \mathcal{A}\},
\end{align}
where the $\alpha_A^k$ is a learnable parameter that will be optimized in the training. The final relevance measure function $S$ is formulated as
\begin{equation}
    \setlength\abovedisplayskip{2pt}
    \setlength\belowdisplayskip{2pt}
    S(v_i,v_j)=\sigma(\Braket{h_i, h_j}), h_i,h_j \in H,
    \label{eq:sim}
\end{equation}
where $H=\{H_A|A\in \mathcal{A}\}$, and $\sigma$ denotes the sigmoid function. 

To optimize the parameters in \ourmethod, we adopt the framework of supervised contrast learning \cite{khosla2020supervised}. We assume that nodes with the same labeled class should be relevant to each other, and vice versa. By contrasting nodes within the same class and different classes, we can force relevant nodes closer to each other, while pushing away the irrelevant nodes. This is formulated as
\begin{equation}
    \setlength\abovedisplayskip{2pt}
    \setlength\belowdisplayskip{2pt}
    \mathcal{L}_S=-\sum_{v_i\in \mathcal{V}}log \frac{\mathbb{E}\sum_{v_j\in I(v_i)}S(v_i,v_j)} {\mathbb{E}\sum_{v_j'\notin I(v_i)}1-S(v_i,v_j')},
\end{equation}
where $I(v_i)$ denotes a set of nodes sharing the same class of $v_i$. In addition, each node should be relevant to itself \cite{shi2014hetesim}. Thus, we propose the self-maximizing loss, which is formulated as
\begin{equation}
    \setlength\abovedisplayskip{2pt}
    \setlength\belowdisplayskip{2pt}
    \mathcal{L}_U= trace\big(\sigma(H^{\top}H)\big) - \textbf{I}.\label{eq:self-max}
\end{equation}
The overall loss is formulated as
\begin{equation}
    \setlength\abovedisplayskip{2pt}
    \setlength\belowdisplayskip{2pt}
    \label{eq:loss}
    \mathcal{L}=\mathcal{L}_S+\mathcal{L}_U.
\end{equation}

\subsection{Overall Algorithm}

The overall training process of \ourmethod is shown in Algorithm \ref{alg:gsim}. Given a heterogeneous graph, we first randomly initialize the node embeddings $Z$ (line 1), and assign all the type-length attention scores $\alpha^k_A$ to 1 (line 2). Then, we use a while-loop to finish the training process (lines 4-11). Specifically, we first adopt the K-layer CP-GNN+ to generate the context information vectors of different lengths $C^k$ (lines 5-7). After that, we use the type-length attention to summarize them into the final node embedding $H$ (lines 8-9). Finally, we use the relevance measure function in Equation \ref{eq:sim} to measure the relevance between nodes (line 10), and the parameters in GSim will be optimized by Equation \ref{eq:loss} (line 11).

In \ourmethod, the computation complexity of each layer GNN is $\mathcal{O}(4|\mathcal{E}|)$. Thus, the overall computation complexity is $\mathcal{O}(4K|\mathcal{E}|)$, where $K$ is the maximum length of context path. According to the statistic of average shortest paths in real-world networks \cite{ye2010distance}, the $K$ can be chosen between 4 to 6 empirically. Thanks to the relation message passing, the storage complexity of the graph is still $O(|\mathcal{V}|+|\mathcal{E}|)$.
\begin{algorithm}[tbp]
    \caption{The training process of \ourmethod}\label{alg:gsim}
    \KwIn {heterogeneous graph $\mathcal{H}=\{\mathcal{V},\mathcal{E},\mathcal{A},\mathcal{R}\}$, Maximum length $K$. Max epoch $ME$}
    \KwOut {Measurement function $S$}
        Randomly initialize the $Z$\;
        Initialize $\alpha_A^k\gets1, k\in [1,K], A \in \mathcal{A}$\;
        Epoch = 0\;
        \While{Epoch $<$ $ME$}{
            \For {$k=1,\ldots,K$}{
                $C^k = \text{k-layer CP-GNN+}(Z)$\;
            }
        $H_A=\sum_{k=1}^K \alpha_A^k C_A^k, A\in \mathcal{A}$\;
        $H=\{H_A|A\in \mathcal{A}\}$\;
        $S(v_i,v_j)=\sigma(\Braket{h_i, h_j}), h_i,h_j \in H$\;
        Optimize parameters using Equation \ref{eq:loss}\;
        Epoch ++\;
        }
    \Return {$S$}
\end{algorithm}

%% file: sections/experiment.tex
\section{Experiment}\label{sec:experiments}
To evaluate the performance of our method, we conduct experiments on four real-world heterogeneous graphs with two downstream tasks: relevance search and community detection.

\subsection{Dataset}
We choose three widely used real-world heterogeneous graph datasets, i.e., ACM \cite{wang2019heterogeneous}, DBLP \cite{gao2009graph}, IMDB \cite{wu2016explaining}, together with a schema-rich knowledge graph dataset AIFB \cite{ristoski2016collection} to evaluate the performance of \ourmethod and other baseline models.
We report their statistics in Table \ref{tab:dataset}, and discuss their details as follows.
\begin{itemize}

    \item ACM dataset \cite{wang2019heterogeneous} is a bibliographic information network with four types of nodes (i.e., Paper, Author, Subject, and Facility). The paper nodes are categorized into 3 classes, i.e., \textit{database}, \textit{wireless communication} and \textit{data mining}.

    \item DBLP dataset \cite{gao2009graph} is a monthly updated citation network consisting of four node types. Nodes are labeled with four classes, i.e., \textit{database} (DB), \textit{data mining} (DM), \textit{information retrieval} (IR) and \textit{machine learning} (ML).

    \item IMDB dataset \cite{wu2016explaining}  consists of four types of nodes. The movie nodes are labeled with three classes, i.e., \textit{Action}, \textit{Comedy}, and \textit{Drama}.

    \item AIFB dataset \cite{ristoski2016collection} is a knowledge graph dataset consisting of 7 types of nodes and 104 types of edges. We choose the ``Personen'' node that is labeled with four classes.  Due to the complexity of the graph, we do not provide detail illustration in Table \ref{tab:dataset}, and pre-define the meta-paths by ourselves.
\end{itemize}

Because only the DBLP dataset provides labels for different types of nodes, we form two datasets (i.e., DBLP-A and DBLP-Multi) where we measure the relevance on author-only and multi-type nodes. For other datasets, we only use one type of labeled nodes for relevance measure.

We adopt the meta-paths in ACM, DBLP, and IMDP defined by previous works to evaluate the meta-path-based methods. Due to the rich-schema property of AIFB, we do not define the meta-path by ourselves. Therefore, some meta-path-based methods cannot be evaluated on the AIFB dataset. Besides, since the unsupervised methods do not need the training data, to make a fair comparison between the unsupervised and supervised methods, the splits of the labeled nodes with 25\%/25\%/50\% for the training, validation, and testing.

\begin{table}[]
    \caption{Statistics of datasets. The labeled node type are highlighted by $*$.}
    \label{tab:dataset}
    \Huge
    \resizebox{\columnwidth}{!}{
        \begin{tabular}{@{}cccccc@{}}
            \toprule
            Dataset & Node type                          & \# Nodes    & Edge type           & \# Edges     & Meta-path \\
            \midrule
            ACM     & \tabincell{c}{Paper* (\textbf{P})                                                                 \\ Author (\textbf{A})\\ Subject (\textbf{S}) \\ Facility (\textbf{F})} & \tabincell{c}{12,499\\ 17,431\\ 73 \\ 1,804} & \tabincell{c}{Paper - Paper \\ Paper - Author \\ Paper - Subject \\ Author - Facility} & \tabincell{c}{30,789 \\ 37,055 \\ 12,499 \\ 30,424} & \tabincell{c}{PAP \\ PSP} \\
            \midrule
            DBLP    & \tabincell{c}{Author* (\textbf{A})                                                                \\ Paper* (\textbf{P}) \\ Conference* (\textbf{C}) \\ Term (\textbf{T}) } & \tabincell{c}{14,475 \\ 14,736 \\ 20 \\ 8,920} & \tabincell{c}{Author - Paper \\ Paper - Conference \\ Paper - Term} & \tabincell{c}{41,794\\ 14,736 \\ 114,624} & \tabincell{c}{APA \\ APCPA \\ APTPA} \\
            \midrule
            IMDB    & \tabincell{c}{Movie* (\textbf{M})                                                                 \\ Actor (\textbf{A})\\ Director (\textbf{D}) \\ Keyword (\textbf{K})} & \tabincell{c}{4,275 \\ 5,432\\ 2,083 \\ 7,313} & \tabincell{c}{Movie - Actor \\ Movie - Director \\ Movie - Keyword} & \tabincell{c}{12,831 \\ 4,181 \\ 20,428} & \tabincell{c}{MAM \\ MDM \\ MKM} \\
            \midrule
            AIFB    & 7 different types                  & Total 7,262 & 104 different types & Total 48,810 & -         \\
            \bottomrule
        \end{tabular}
    }
\end{table}

\subsection{Baselines}
To evaluate the effectiveness of our measure, we compare it with a list of the state-of-the-art node embedding methods (i.e., Node2vec \cite{grover2016node2vec}, Metapath2vec \cite{dong2017metapath2vec}, and HIN2vec \cite{fu2017hin2vec}) and GNN-based methods (i.e., GCN \cite{kipf2017semisupervised}, GAT \cite{velivckovic2017graph}, LGNN \cite{chen2019supervised}, HAN \cite{wang2019heterogeneous}, HGT \cite{hu2020heterogeneous}, and CP-GNN \cite{luo2021detecting}). Especially, HGT is a semi-supervised neural network model which adopts the transformer mechanism to differentiate the importance of relations. CP-GNN is the earlier version of our work that utilizes the context path to excavate semantics in heterogeneous graphs. We adopt Equation \ref{eq:sim} on the learned embeddings to calculate the relevance.

Except for graph embedding-based and GNN-based baselines, we also choose two traditional relevance measure methods (i.e., SimRank \cite{jeh2002simrank} and HeteSim \cite{shi2014hetesim}) for comparison.

\subsection{Parameters Settings}
We now briefly discuss the settings of model parameters.
For graph embedding-based approaches such as Node2vec and Metapath2vec, we respectively set the length of random walk to 100, the sampling window size to 5, the number of walks per node to 120, and the number of negative samplings to 5.
For GNN-based methods such as GCN, GAT, LGNN, HAN, and HGT, the number of graph convolution layers is set to 2, and their node features are first randomly initialized, and then updated during the model learning process. The dimension of node feature embedding for all compared methods is set to 128. We set these parameters following previous research \cite{dong2017metapath2vec,wang2019heterogeneous,hu2020heterogeneous}.

For our \ourmethod, the number of attention heads is set to 2, the dimension of the output vectors of K/Q-Linear components is set to 128, and the node dropout rate is set to 0.3. The maximum $K$ is set to 4. The Adam \cite{kingma2014adam} is adopted to optimize all models, and the learning rate is set to 0.05. We analyze the impact of these parameters in section \ref{sec:parameters}.

\subsection{Relevance Search Results}
Another application of relevance measure is the relevance search. In relevance search, given a query node, we want to find its top-$N$ relevant nodes. We adopt the average recall as the metric, which is computed as
\begin{equation}
    \setlength\abovedisplayskip{1pt}
    \setlength\belowdisplayskip{1pt}
    recall@N = \frac{|r_i\in R\wedge f_I(r_i)=f_I(q)|}{N},
\end{equation}
where $q$ denotes the query node, $f_I(\cdot)$ denotes the class label of node, and $R$ denotes the set of top-$N$ relevant results.
In experiments, we randomly select 50 query nodes for each dataset, then we return top-10 nodes based on the measurement results.
The experiment results are shown in Table \ref{tab:avg-recall}.

From the results in Table \ref{tab:avg-recall}, we can find that \ourmethod outperforms other baselines on all datasets. However, comparing the GNN-based methods with the node embedding methods and conventional methods, we can see that GNN-based methods achieve inferior performance in the relevance search task. The possible reason could be that the GNN-based methods suffer from the over-smoothing problem \cite{chen2020measuring}. This means that the final node representations tend to converge to the same value, making them hard to be distinguished. Therefore, the relative relevance between nodes is erased, resulting in weak performance in relevance search.

As for \ourmethod, it adopts the supervised contrast learning to contrast nodes within the same class and different classes during training. This will force relevant nodes closer to each other while pushing away the irrelevant nodes. In this way, \ourmethod enables us to capture the relative relevance between nodes and achieve the best performance in relevance search.

\begin{table}[tbp]
    \center
    \caption{Performance of different measures on relevance search task.}
    \label{tab:avg-recall}
    \resizebox{0.9\columnwidth}{!}{%
        \begin{tabular}{@{}c|ccccc@{}}
            \toprule
            \multirow{2}{*}{Method} & ACM                        & DBLP-A         & DBLP-Multi     & \multicolumn{1}{l}{IMDB} & AIFB           \\ \cmidrule(l){2-6}
                                    & \multicolumn{5}{c}{Recall}                                                                               \\ \midrule
            SimRank                 & 0.450                      & {\ul 0.826}    & {\ul 0.850}    & 0.430                    & 0.412          \\
            Hetesim                 & {\ul 0.646}                & 0.358          & 0.367          & 0.410                    & -              \\ \midrule
            Node2Vec                & 0.450                      & 0.556          & 0.497          & {\ul 0.520}              & 0.330          \\
            Metapath2vec            & 0.468                      & 0.292          & 0.751          & 0.484                    & -              \\ \midrule
            GCN                     & 0.382                      & 0.312          & 0.354          & 0.382                    & 0.455          \\
            GAT                     & 0.436                      & 0.470          & 0.307          & 0.332                    & 0.396          \\
            LGNN                    & 0.528                      & 0.272          & 0.293          & 0.278                    & 0.594          \\
            HAN                     & 0.534                      & 0.286          & -              & 0.512                    & -              \\
            HGT                     & 0.504                      & 0.472          & 0.710          & 0.388                    & {\ul 0.626}    \\ \midrule
            CP-GNN                  & {\ul 0.646}                & 0.644          & -              & 0.454                    & 0.356          \\
            \ourmethod              & \textbf{0.782}             & \textbf{0.900} & \textbf{0.888} & \textbf{0.524}           & \textbf{0.664} \\ \bottomrule
        \end{tabular}%
    }
\end{table}

\subsubsection{Case Study 1: Find same type relevant nodes}
In this study, we illustrate the top-10 nodes of the same type as the query node. We first query the top-10 relevant movies for the movie ``Twilight'' in IMDB dataset, which is one of the most famous drama movies. The results are shown in Table \ref{tab:recall-imdb}. We respectively list the names, relevance scores, and labels of each movie found.

From the results, we can see that the results returned by \ourmethod are better than the other methods, of which the "Drama" movies are the most (i.e., 8/10). On the contrary, other baselines contain many ``Adventure'' and ``Action'' movies in their results. Specifically, SimRank returns the most diverse results containing movies of different genera. HeteSim returns ``The Twilight Sega: New Moon'' and ``The Twilight Saga: Eclipse'', which are the sequels to ``Twilight''. The possible reason is that these movies share similar actors, and such relations are captured by the MAM meta-path. But these actors might also act in movies of other genera, which deteriorates the search results. HGT tries to automatically discover the semantics, but the found semantics might not be suitable for the query node. Thus, it returns lots of ``Adventure'' movies and the results are not very distinctive. \ourmethod adopts the context path to excavate the semantics that is crucial to the query node. For example, the 4th and 6th results (``Blood and Chocolate'' and ``Blood Ties'') both contain horror and action scenes which are also the major themes of ``Twilight''. These underlying relations cannot be well revealed by the meta-path.

In Table \ref{tab:recall-dblp}, we try to find relevant authors for author ``Weidong Chen'' who is a researcher in the database area (DB). From the results, we can see that all the authors found by \ourmethod are in DB area, whereas other baselines return some authors in the information retrieval (IR), data mining (DM), and artificial intelligence (AI) areas. In addition, we can see that the first result returned by the HGT is not ``Weidong Chen''. This shows that HGT does not satisfy the self-maximizing property which is important in previous methods (SimRank and HeteSim). Thanks to the self-maximizing loss shown in Equation \ref{eq:self-max}, \ourmethod enables to make sure that each node is most relevant to itself.

\begin{table*}[tbp]
    \caption{Top-10 query results for movie:``Twilight" (label: Drama) on IMDB dataset. }
    \label{tab:recall-imdb}
    \resizebox{\textwidth}{!}{%
        \begin{tabular}{@{}c|ccc|ccc|ccc|ccc@{}}
            \toprule
            \multirow{2}{*}{Rank} & \multicolumn{3}{c|}{SimRank}                              & \multicolumn{3}{c|}{HeteSim (MAM)} & \multicolumn{3}{c|}{HGT} & \multicolumn{3}{c}{\ourmethod}                                                                                                                                                                                                                                                   \\ \cmidrule(l){2-13}
                                  & \multicolumn{1}{c|}{Movie}                                & \multicolumn{1}{c|}{Score}         & Label                    & \multicolumn{1}{c|}{Movie}                                & \multicolumn{1}{c|}{Score} & Label     & \multicolumn{1}{c|}{Author}          & \multicolumn{1}{c|}{Score} & Label     & \multicolumn{1}{c|}{Movie}                         & \multicolumn{1}{c|}{Score} & Label     \\ \midrule
            1                     & \multicolumn{1}{c|}{\textbf{Twilight}}                    & \multicolumn{1}{c|}{0.148}         & Drama                    & \multicolumn{1}{c|}{\textbf{Twilight}}                    & \multicolumn{1}{c|}{1.000} & Drama     & \multicolumn{1}{c|}{Trash}           & \multicolumn{1}{c|}{1.000} & Drama     & \multicolumn{1}{c|}{\textbf{Twilight}}             & \multicolumn{1}{c|}{1.000} & Drama     \\
            2                     & \multicolumn{1}{c|}{We Need to Talk About Kevin}          & \multicolumn{1}{c|}{0.055}         & Drama                    & \multicolumn{1}{c|}{The Twilight Saga: New Moon}          & \multicolumn{1}{c|}{0.667} & Drama     & \multicolumn{1}{c|}{Modern Problems} & \multicolumn{1}{c|}{1.000} & Adventure & \multicolumn{1}{c|}{Shallow Hal}                   & \multicolumn{1}{c|}{1.000} & Adventure \\
            3                     & \multicolumn{1}{c|}{Gangster Squad}                       & \multicolumn{1}{c|}{0.026}         & Action                   & \multicolumn{1}{c|}{The Twilight Saga: Eclipse}           & \multicolumn{1}{c|}{0.667} & Drama     & \multicolumn{1}{c|}{Unleashed}       & \multicolumn{1}{c|}{1.000} & Action    & \multicolumn{1}{c|}{The Bridges of Madison County} & \multicolumn{1}{c|}{1.000} & Drama     \\
            4                     & \multicolumn{1}{c|}{A Thin Line Between Love and   Hate}  & \multicolumn{1}{c|}{0.010}         & Adventure                & \multicolumn{1}{c|}{Zathura: A Space Adventure}           & \multicolumn{1}{c|}{0.333} & Action    & \multicolumn{1}{c|}{The Other Woman} & \multicolumn{1}{c|}{1.000} & Adventure & \multicolumn{1}{c|}{Blood and Chocolate}           & \multicolumn{1}{c|}{1.000} & Drama     \\
            5                     & \multicolumn{1}{c|}{Get Carter}                           & \multicolumn{1}{c|}{0.007}         & Action                   & \multicolumn{1}{c|}{What to Expect When You're Expecting} & \multicolumn{1}{c|}{0.333} & Adventure & \multicolumn{1}{c|}{Pain \& Gain}    & \multicolumn{1}{c|}{1.000} & Adventure & \multicolumn{1}{c|}{A Walk to Remember}            & \multicolumn{1}{c|}{1.000} & Drama     \\
            6                     & \multicolumn{1}{c|}{Girls Gone Dead}                      & \multicolumn{1}{c|}{0.005}         & Adventure                & \multicolumn{1}{c|}{Drinking Buddies}                     & \multicolumn{1}{c|}{0.333} & Adventure & \multicolumn{1}{c|}{Rio}             & \multicolumn{1}{c|}{1.000} & Adventure & \multicolumn{1}{c|}{Blood Ties}                    & \multicolumn{1}{c|}{1.000} & Drama     \\
            7                     & \multicolumn{1}{c|}{Jab Tak Hai Jaan}                     & \multicolumn{1}{c|}{0.005}         & Drama                    & \multicolumn{1}{c|}{The Ridiculous 6}                     & \multicolumn{1}{c|}{0.333} & Adventure & \multicolumn{1}{c|}{Guess Who}       & \multicolumn{1}{c|}{1.000} & Adventure & \multicolumn{1}{c|}{The Brothers Bloom}            & \multicolumn{1}{c|}{1.000} & Adventure \\
            8                     & \multicolumn{1}{c|}{Red Dog}                              & \multicolumn{1}{c|}{0.005}         & Adventure                & \multicolumn{1}{c|}{Welcome to the Rileys}                & \multicolumn{1}{c|}{0.333} & Drama     & \multicolumn{1}{c|}{Faithful}        & \multicolumn{1}{c|}{1.000} & Adventure & \multicolumn{1}{c|}{Titanic}                       & \multicolumn{1}{c|}{1.000} & Drama     \\
            9                     & \multicolumn{1}{c|}{Captain America: The First   Avenger} & \multicolumn{1}{c|}{0.004}         & Action                   & \multicolumn{1}{c|}{The Last Five Years}                  & \multicolumn{1}{c|}{0.333} & Adventure & \multicolumn{1}{c|}{Inside Out}      & \multicolumn{1}{c|}{1.000} & Adventure & \multicolumn{1}{c|}{Womb}                          & \multicolumn{1}{c|}{1.000} & Drama     \\
            10                    & \multicolumn{1}{c|}{Ironclad}                             & \multicolumn{1}{c|}{0.004}         & Action                   & \multicolumn{1}{c|}{On the Road}                          & \multicolumn{1}{c|}{0.333} & Drama     & \multicolumn{1}{c|}{Death Sentence}  & \multicolumn{1}{c|}{1.000} & Action    & \multicolumn{1}{c|}{Code 46}                       & \multicolumn{1}{c|}{1.000} & Drama     \\ \bottomrule
        \end{tabular}%
    }
\end{table*}

\begin{table*}[tbp]
    \caption{Top-10 query results for author: ``Weidong Chen'' (label: DB) on DBLP-A dataset.}
    \label{tab:recall-dblp}
    \resizebox{\textwidth}{!}{%
        \begin{tabular}{@{}c|ccc|ccc|ccc|ccc@{}}
            \toprule
            \multirow{2}{*}{Rank} & \multicolumn{3}{c|}{SimRank}               & \multicolumn{3}{c|}{HeteSim (APA)} & \multicolumn{3}{c|}{HGT} & \multicolumn{3}{c}{\ourmethod}                                                                                                                                                                                                                   \\ \cmidrule(l){2-13}
                                  & \multicolumn{1}{c|}{Author}                & \multicolumn{1}{c|}{Score}         & Label                    & \multicolumn{1}{c|}{Author}                & \multicolumn{1}{c|}{Score} & Label & \multicolumn{1}{c|}{Author}             & \multicolumn{1}{c|}{Score} & Label & \multicolumn{1}{c|}{Author}                & \multicolumn{1}{c|}{Score} & Label \\ \midrule
            1                     & \multicolumn{1}{c|}{\textbf{Weidong Chen}} & \multicolumn{1}{c|}{0.129}         & DB                       & \multicolumn{1}{c|}{\textbf{Weidong Chen}} & \multicolumn{1}{c|}{1.000} & DB    & \multicolumn{1}{c|}{Matthew Denny}      & \multicolumn{1}{c|}{1.000} & DB    & \multicolumn{1}{c|}{\textbf{Weidong Chen}} & \multicolumn{1}{c|}{0.982} & DB    \\
            2                     & \multicolumn{1}{c|}{Serge Rielau}          & \multicolumn{1}{c|}{0.007}         & DB                       & \multicolumn{1}{c|}{Serge Rielau}          & \multicolumn{1}{c|}{0.224} & DB    & \multicolumn{1}{c|}{Danette Chimenti}   & \multicolumn{1}{c|}{1.000} & DB    & \multicolumn{1}{c|}{Kurt Ingenthron}       & \multicolumn{1}{c|}{0.926} & DB    \\
            3                     & \multicolumn{1}{c|}{Ben Hutchinson}        & \multicolumn{1}{c|}{0.002}         & AI                       & \multicolumn{1}{c|}{Eliezer Levy}          & \multicolumn{1}{c|}{0.000} & DB    & \multicolumn{1}{c|}{Christian Zimmer}   & \multicolumn{1}{c|}{1.000} & IR    & \multicolumn{1}{c|}{Adam Silberstein}      & \multicolumn{1}{c|}{0.924} & DB    \\
            4                     & \multicolumn{1}{c|}{Eckhard D. Falkenberg} & \multicolumn{1}{c|}{0.002}         & DB                       & \multicolumn{1}{c|}{Richard Sidle}         & \multicolumn{1}{c|}{0.000} & DB    & \multicolumn{1}{c|}{Volker Linnemann}   & \multicolumn{1}{c|}{1.000} & DB    & \multicolumn{1}{c|}{David E. Bakkom}       & \multicolumn{1}{c|}{0.924} & DB    \\
            5                     & \multicolumn{1}{c|}{Hal R. Varian}         & \multicolumn{1}{c|}{0.002}         & IR                       & \multicolumn{1}{c|}{Roger Nasr}            & \multicolumn{1}{c|}{0.000} & IR    & \multicolumn{1}{c|}{Mihalis Yannakakis} & \multicolumn{1}{c|}{1.000} & DB    & \multicolumn{1}{c|}{Xin Luna Dong}         & \multicolumn{1}{c|}{0.923} & DB    \\
            6                     & \multicolumn{1}{c|}{Balder ten Cate}       & \multicolumn{1}{c|}{0.001}         & DB                       & \multicolumn{1}{c|}{Keizo Oyama}           & \multicolumn{1}{c|}{0.000} & IR    & \multicolumn{1}{c|}{Stephen C. North}   & \multicolumn{1}{c|}{1.000} & DM    & \multicolumn{1}{c|}{Pavel Avgustinov}      & \multicolumn{1}{c|}{0.923} & DB    \\
            7                     & \multicolumn{1}{c|}{Bennet Vance}          & \multicolumn{1}{c|}{0.001}         & DB                       & \multicolumn{1}{c|}{Hai Leong Chieu}       & \multicolumn{1}{c|}{0.000} & IR    & \multicolumn{1}{c|}{Cyril Goutte}       & \multicolumn{1}{c|}{1.000} & IR    & \multicolumn{1}{c|}{Jen-Yao Chung}         & \multicolumn{1}{c|}{0.921} & DB    \\
            8                     & \multicolumn{1}{c|}{Stephen J. Hegner}     & \multicolumn{1}{c|}{0.001}         & DB                       & \multicolumn{1}{c|}{Mark Coyle}            & \multicolumn{1}{c|}{0.000} & DB    & \multicolumn{1}{c|}{Armin B. Cremers}   & \multicolumn{1}{c|}{1.000} & IR    & \multicolumn{1}{c|}{H. V. Jagadish}        & \multicolumn{1}{c|}{0.921} & DB    \\
            9                     & \multicolumn{1}{c|}{Amit Ramesh}           & \multicolumn{1}{c|}{0.001}         & AI                       & \multicolumn{1}{c|}{Fabrizio Angiulli}     & \multicolumn{1}{c|}{0.000} & DM    & \multicolumn{1}{c|}{Bruno Defude}       & \multicolumn{1}{c|}{1.000} & IR    & \multicolumn{1}{c|}{Mary Tork Roth}        & \multicolumn{1}{c|}{0.920} & DB    \\
            10                    & \multicolumn{1}{c|}{Marc Spielmann}        & \multicolumn{1}{c|}{0.001}         & DB                       & \multicolumn{1}{c|}{Denver Dash}           & \multicolumn{1}{c|}{0.000} & AI    & \multicolumn{1}{c|}{Shanshan Wang}      & \multicolumn{1}{c|}{1.000} & DM    & \multicolumn{1}{c|}{Nancy D. Griffeth}     & \multicolumn{1}{c|}{0.920} & DB    \\ \bottomrule
        \end{tabular}%
    }
\end{table*}

\begin{table*}[tbp]
    \caption{Top-5 multi-type query results for author: ``Weidong Chen'' (label: DB) on DBLP-Multi dataset. (Due to the limitation of space, we replace the actual paper names with P1-P5.)}
    \label{tab:recall-DBLP-multi}
    \resizebox{\textwidth}{!}{%
        \begin{tabular}{@{}c|cccccc|cccccc@{}}
            \toprule
            Method & \multicolumn{6}{c|}{SimRank}    & \multicolumn{6}{c}{HeteSim}                                                                                                                                                                                                                                                                                                         \\ \midrule
            Rank   & \multicolumn{1}{c|}{Conference} & \multicolumn{1}{c|}{Label}     & \multicolumn{1}{c|}{Author}                & \multicolumn{1}{c|}{Label} & \multicolumn{1}{c|}{Paper} & Label & \multicolumn{1}{c|}{Conference (APC)} & \multicolumn{1}{c|}{Label} & \multicolumn{1}{c|}{Author (APA)}     & \multicolumn{1}{c|}{Label} & \multicolumn{1}{c|}{Paper (APAP)} & Label \\ \midrule
            1      & \multicolumn{1}{c|}{SDM}        & \multicolumn{1}{c|}{DM}        & \multicolumn{1}{c|}{Weidong Chen}          & \multicolumn{1}{c|}{DB}    & \multicolumn{1}{c|}{P1}    & DB    & \multicolumn{1}{c|}{VLDB}             & \multicolumn{1}{c|}{DB}    & \multicolumn{1}{c|}{Weidong Chen}     & \multicolumn{1}{c|}{DB}    & \multicolumn{1}{c|}{P1}           & DB    \\
            2      & \multicolumn{1}{c|}{EDBT}       & \multicolumn{1}{c|}{DB}        & \multicolumn{1}{c|}{Serge Rielau}          & \multicolumn{1}{c|}{DB}    & \multicolumn{1}{c|}{P2}    & DB    & \multicolumn{1}{c|}{SDM}              & \multicolumn{1}{c|}{DM}    & \multicolumn{1}{c|}{Serge Rielau}     & \multicolumn{1}{c|}{DB}    & \multicolumn{1}{c|}{P2}           & DB    \\
            3      & \multicolumn{1}{c|}{ECML}       & \multicolumn{1}{c|}{AI}        & \multicolumn{1}{c|}{Ben Hutchinson}        & \multicolumn{1}{c|}{AI}    & \multicolumn{1}{c|}{P3}    & IR    & \multicolumn{1}{c|}{EDBT}             & \multicolumn{1}{c|}{DB}    & \multicolumn{1}{c|}{Eliezer Levy}     & \multicolumn{1}{c|}{DB}    & \multicolumn{1}{c|}{P3}           & IR    \\
            4      & \multicolumn{1}{c|}{VLDB}       & \multicolumn{1}{c|}{DB}        & \multicolumn{1}{c|}{Eckhard D. Falkenberg} & \multicolumn{1}{c|}{DB}    & \multicolumn{1}{c|}{P4}    & DB    & \multicolumn{1}{c|}{ECML}             & \multicolumn{1}{c|}{AI}    & \multicolumn{1}{c|}{Richard Sidle}    & \multicolumn{1}{c|}{DB}    & \multicolumn{1}{c|}{P4}           & DB    \\
            5      & \multicolumn{1}{c|}{ICML}       & \multicolumn{1}{c|}{AI}        & \multicolumn{1}{c|}{Hal R. Varian}         & \multicolumn{1}{c|}{IR}    & \multicolumn{1}{c|}{P5}    & DM    & \multicolumn{1}{c|}{ICML}             & \multicolumn{1}{c|}{AI}    & \multicolumn{1}{c|}{Roger Nasr}       & \multicolumn{1}{c|}{IR}    & \multicolumn{1}{c|}{P5}           & DM    \\ \midrule
            Method & \multicolumn{6}{c|}{HGT}        & \multicolumn{6}{c}{\ourmethod}                                                                                                                                                                                                                                                                                                      \\ \midrule
            Rank   & \multicolumn{1}{c|}{Conference} & \multicolumn{1}{c|}{Label}     & \multicolumn{1}{c|}{Author}                & \multicolumn{1}{c|}{Label} & \multicolumn{1}{c|}{Paper} & Label & \multicolumn{1}{c|}{Conference}       & \multicolumn{1}{c|}{Label} & \multicolumn{1}{c|}{Author}           & \multicolumn{1}{c|}{Label} & \multicolumn{1}{c|}{Paper}        & Label \\ \midrule
            1      & \multicolumn{1}{c|}{VLDB}       & \multicolumn{1}{c|}{DB}        & \multicolumn{1}{c|}{Matthew Denny}         & \multicolumn{1}{c|}{DB}    & \multicolumn{1}{c|}{P1}    & DB    & \multicolumn{1}{c|}{WSDM}             & \multicolumn{1}{c|}{IR}    & \multicolumn{1}{c|}{Weidong Chen}     & \multicolumn{1}{c|}{DB}    & \multicolumn{1}{c|}{P1}           & DB    \\
            2      & \multicolumn{1}{c|}{SDM}        & \multicolumn{1}{c|}{DM}        & \multicolumn{1}{c|}{Danette Chimenti}      & \multicolumn{1}{c|}{DB}    & \multicolumn{1}{c|}{P2}    & DB    & \multicolumn{1}{c|}{ICML}             & \multicolumn{1}{c|}{AI}    & \multicolumn{1}{c|}{Kurt Ingenthron}  & \multicolumn{1}{c|}{DB}    & \multicolumn{1}{c|}{P2}           & DB    \\
            3      & \multicolumn{1}{c|}{EDBT}       & \multicolumn{1}{c|}{DB}        & \multicolumn{1}{c|}{Christian Zimmer}      & \multicolumn{1}{c|}{IR}    & \multicolumn{1}{c|}{P3}    & IR    & \multicolumn{1}{c|}{EDBT}             & \multicolumn{1}{c|}{DB}    & \multicolumn{1}{c|}{Adam Silberstein} & \multicolumn{1}{c|}{DB}    & \multicolumn{1}{c|}{P3}           & DB    \\
            4      & \multicolumn{1}{c|}{ECML}       & \multicolumn{1}{c|}{AI}        & \multicolumn{1}{c|}{Volker Linnemann}      & \multicolumn{1}{c|}{DB}    & \multicolumn{1}{c|}{P4}    & DB    & \multicolumn{1}{c|}{SDM}              & \multicolumn{1}{c|}{DM}    & \multicolumn{1}{c|}{David E. Bakkom}  & \multicolumn{1}{c|}{DB}    & \multicolumn{1}{c|}{P4}           & DB    \\
            5      & \multicolumn{1}{c|}{ICML}       & \multicolumn{1}{c|}{AI}        & \multicolumn{1}{c|}{Mihalis Yannakakis}    & \multicolumn{1}{c|}{DB}    & \multicolumn{1}{c|}{P5}    & DM    & \multicolumn{1}{c|}{SIGIR}            & \multicolumn{1}{c|}{IR}    & \multicolumn{1}{c|}{Xin Luna Dong}    & \multicolumn{1}{c|}{DB}    & \multicolumn{1}{c|}{P5}           & DB    \\ \bottomrule
        \end{tabular}%
    }
\end{table*}

\subsubsection{Case Study 2: Find different types relevant nodes}
In relevance measure, we need to evaluate the relevance of different types of nodes. Therefore, in DBLP-Multi dataset, given an author ``Weidong Chen'', we want to find the relevant conferences, authors, and papers for him. The top-5 results of different types of nodes are shown in Table \ref{tab:recall-DBLP-multi}. Due to the limitation of space, we replace the actual paper names with P1-P5.

From the results in Table \ref{tab:recall-DBLP-multi}, we can see that SimRank still achieves the worst results. The reason is that SimRank disregards the node type in heterogeneous graphs. Thus, it cannot find relations between different types of nodes. HeteSim behaves better than SimRank. But it requires defining different meta-paths for different types of nodes. It is infeasible when the number of node types increases. On the contrary, HGT can find the relationship between different types of nodes when calculating their relevance. But it fails to capture the asymmetric semantics, thus performing poorly for the relevance measure between \textit{Author-Paper}. \ourmethod adopts the relation message passing, which could comprehensively capture both the asymmetric and symmetric semantics and reach great results in relevance measures between \textit{Author-Paper} and \textit{Author-Author}. However, \ourmethod fails on the \textit{Author-Conference}. The possible reason is that there are only 5 labeled conference nodes in DBLP-Multi dataset for training. Due to the limited number of training nodes, \ourmethod cannot fully learn the relevance of conference nodes.

\begin{table*}[]
    \caption{Performance of different measures on community detection task.}
    \label{tab:community}
    \resizebox{\textwidth}{!}{%
        \begin{tabular}{@{}c|cccc|cccc|cccc|cccc|cccc@{}}
            \toprule
            \multirow{2}{*}{Method} & \multicolumn{4}{c|}{ACM} & \multicolumn{4}{c|}{DBLP-A} & \multicolumn{4}{c|}{DBLP-Multi} & \multicolumn{4}{c|}{IMDB} & \multicolumn{4}{c}{AIFB}                                                                                                                                                                                                                                                                                                  \\ \cmidrule(l){2-21}
                                    & F-score                  & NMI                         & ARI                             & Purity                    & F-score                  & NMI             & ARI             & Purity          & F-score                & NMI             & ARI             & Purity          & F-score         & NMI                   & ARI             & Purity          & F-score               & NMI             & ARI             & Purity          \\ \midrule
            SimRank                 & 0.5395                   & 0.0033                      & 0.0003                          & 0.4908                    & 0.5055                   & 0.3798          & 0.2125          & 0.4993          & 0.404                  & 0.0085          & 0               & 0.3078          & \textbf{0.5101} & 0.0034                & -0.0003         & 0.3803          & 0.3457                & 0.0648          & -0.004          & 0.4773          \\
            HeteSim                 & 0.5353                   & 0.0038                      & 0.003                           & 0.4933                    & 0.3978                   & 0.0058          & -0.002          & 0.3075          & 0.4092                 & 0.0083          & 0.0011          & 0.3030          & {\ul 0.5070}    & 0.0015                & -0.0003         & 0.3804          & \multicolumn{4}{c}{-}                                                       \\ \midrule
            Node2Vec                & 0.3506                   & 0.1324                      & 0.1056                          & 0.5413                    & 0.3196                   & 0.0972          & 0.08            & 0.4352          & 0.3164                 & 0.0795          & 0.0742          & 0.4313          & 0.3938          & 0.0777                & 0.0817          & 0.4920          & 0.3834                & 0.2011          & 0.1863          & 0.6013          \\
            Metapath2vec            & 0.4397                   & 0.1586                      & 0.1276                          & 0.5718                    & 0.3093                   & 0.0739          & 0.0484          & 0.3982          & 0.3536                 & 0.2834          & 0.2672          & 0.4012          & 0.3761          & 0.0219                & 0.0186          & 0.4387          & \multicolumn{4}{c}{-}                                                       \\ \midrule
            GCN                     & 0.4090                   & 0.0375                      & 0.0385                          & 0.5161                    & 0.2616                   & 0.0071          & 0.0048          & 0.3144          & 0.2754                 & 0.0101          & 0.0123          & 0.3236          & 0.3539          & 0.0013                & 0.0006          & 0.3934          & 0.4157                & 0.1444          & 0.1325          & 0.5795          \\
            GAT                     & 0.5595                   & 0.1676                      & 0.0926                          & 0.5693                    & 0.6439                   & 0.5174          & 0.5087          & {\ul 0.7536}    & 0.3945                 & 0.1561          & 0.0925          & 0.4294          & 0.495           & 0.0002                & -0.0004         & 0.3803          & 0.4506                & 0.2286          & {\ul 0.1935}    & 0.6023          \\
            LGNN                    & 0.6243                   & 0.4013                      & 0.4026                          & 0.7139                    & 0.4071                   & 0.0088          & -0.0009         & 0.3016          & 0.3959                 & 0.0073          & -0.0007         & 0.3030          & 0.4854          & 0.0156                & -0.008          & 0.3803          & {\ul 0.4572}          & 0.2437          & 0.1166          & 0.6136          \\
            HAN                     & {\ul 0.6379}             & {\ul0.4062}                 & {\ul 0.4369}                    & {\ul 0.7695}              & 0.2774                   & 0.0028          & 0.0011          & 0.3135          & \multicolumn{4}{c|}{-} & 0.3832          & 0.065           & 0.0678          & 0.5005          & \multicolumn{4}{c}{-}                                                                                                                   \\
            HGT                     & 0.5589                   & 0.2847                      & 0.2632                          & 0.6239                    & {\ul 0.6721}             & {\ul 0.5853}    & {\ul 0.5548}    & 0.7496          & {\ul 0.7797}           & {\ul 0.6754}    & {\ul 0.7025}    & {\ul 0.8679}    & 0.4941          & 0.0013                & 0               & 0.3798          & 0.4471                & {\ul 0.2688}    & 0.1902          & {\ul 0.6250}    \\ \midrule
            CP-GNN                  & 0.6170                   & 0.4199                      & 0.3718                          & 0.7069                    & 0.4507                   & 0.2738          & 0.2597          & 0.6175          & \multicolumn{4}{c|}{-} & 0.405           & {\ul 0.0899}    & {\ul 0.096}     & {\ul 0.5047}    & 0.3836                & 0.1146          & 0.1034          & 0.5341                                                                      \\
            \ourmethod              & \textbf{0.6968}          & \textbf{0.5065}             & \textbf{0.5136}                 & \textbf{0.8058}           & \textbf{0.8820}          & \textbf{0.7857} & \textbf{0.8411} & \textbf{0.9354} & \textbf{0.8620}        & \textbf{0.7539} & \textbf{0.8140} & \textbf{0.9239} & 0.4304          & \textbf{0.1181}       & \textbf{0.1234} & \textbf{0.5697} & \textbf{0.6688}       & \textbf{0.5001} & \textbf{0.4855} & \textbf{0.7386} \\ \bottomrule
        \end{tabular}%
    }
\end{table*}

\begin{figure*}[tbp]
    \centering
    \includegraphics[width=\textwidth]{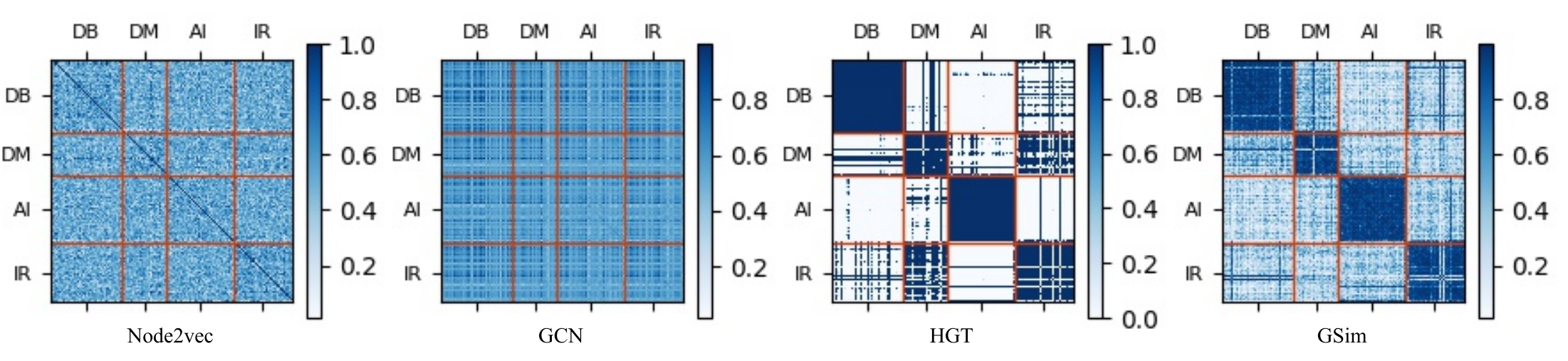}
    \caption{The relevance matrices generated by different methods on DBLP-Multi dataset. (DB: ``Database'', DM: ``Data Mining'', AI: ``Artificial Intelligence'', IR: ``Information retrieval'').}
    \label{fig:vis-matrix}
\end{figure*}

\subsection{Community Detection Results}
Relevance measure plays an essential role in community detection. Following previous research \cite{shi2014hetesim,wang2020effective}, we apply the Spectral Clustering \cite{von2007tutorial} on the relevance scores generated by different methods to perform the community detection. We choose the widely used F-score, NMI, ARI, and Purity as the metrics. The results are reported in Table \ref{tab:community}, where the best results are highlighted in bold and the second-best results are labeled with underlines.

From the results in Table \ref{tab:community}, we can see that \ourmethod outperforms other baselines on most datasets concerning different metrics. This demonstrates that \ourmethod can generate high-quality measure results via automatically excavating semantics in heterogeneous graphs.

As for node embedding methods (i.e., Node2vec and Metapath2vec), they perform better than the conventional measure methods (i.e., SimRank and HeteSim). This indicates that the learned node embedding can successfully consider the structure information in the graph which is essential for the relevance measure. Besides, the Metapath2vec only beats Node2vec in ACM and DBLP-Multi datasets. This shows that even when semantics is considered by adopting meta-paths, the choice of meta-path will also severely influence the performance of Metapath2vec. Additionally, the Metapath2vec uses only one meta-path, which cannot capture the semantics comprehensively.

For GNN-based methods, we can see that the GCN, GAT, and LGNN achieve relatively worse results. The possible reason is that they are originally proposed for homogeneous graphs, thus they do not consider the complex semantics in heterogeneous graphs. GAT performs better than GCN, which reflects the importance of the attention mechanism. The attention mechanism used in GAT can be regarded as a simple way to differentiate the node type and edge type in heterogeneous graphs. Thanks to the meta-path, HAN can explicitly excavate complex semantic information to achieve a better result. But due to the symmetry constraint of meta-path, HAN cannot handle the relevance measure between multiple node types in DBLP-Multi. On the other hand, HGT and CP-GNN do not require the meta-path. They adopt the designed attention mechanism to automatically capture the semantics in graphs, which helps them outperform other GNN-based methods in most datasets. But due to the constraint of context path, CP-GNN cannot be used in DBLP-Multi.

For \ourmethod, it adopts the context path to spontaneously leverage the semantics in the graph and utilizes the relation attention to differentiate their importance. Besides, thanks to the relation message passing, \ourmethod can capture both the symmetric and asymmetric semantics, thus it can measure the relevance between any node type.

\subsubsection{Case Study 3: Relevance matrix visualization}
To intuitively understand the relevance measure results of community detection, we first visualize the relevance matrices generated by different methods on the DBLP-Multi dataset in Fig. \ref{fig:vis-matrix}. Then, we selected several authors with different labels from DBLP-Multi and illustrate the relevance scores between them in Fig. \ref{fig:selected-matrix}.

In Fig. \ref{fig:vis-matrix}, we demonstrate the relevance matrices generated by Node2vec, GCN, HGT, and \ourmethod. Each entry of the relevance matrix depicts the relevance score between two nodes. The darker the color, the more relevant the two nodes are. The nodes in DBLP-Multi can be classified into four classes (i.e., DB: ``Database'', DM: ``Data Mining'', AI: ``Artificial Intelligence'', IR: ``Information retrieval''). Therefore, we group the nodes of the same class and divide them using orange lines.

From the results in Fig. \ref{fig:vis-matrix}, we can see that Node2vec and GCN generate inferior results. The relevance matrix of Node2vec is quite random except for the diagonal entries. This shows that Node2vec fails to consider the complex semantics in heterogeneous graphs and only captures simple structure information. Thus, it only captures the relevance between nodes and themselves. GCN generates a quite smooth relevance matrix, where each node shares a similar relevance score. This is because GCN adopts the graph convolutional layer that would smoothen the node embeddings. Thus, node embedding of each node will become similar and the relevance is measured inaccurately following the embeddings. The results of HGT are better, where nodes of the same class have higher relevant scores. However, it still has some error results. For example, HGT generates high relevance scores for some nodes between IR and DM. \ourmethod achieves relative better results, where nodes within the same classes are highly related whereas nodes between different classes have lower relevance scores.

To better demonstrate the relevance measure results, we select several authors with different labels and use the \ourmethod to measure their relevance scores. In Fig. \ref{fig:selected-matrix}, we can see that each author is most relevant to himself, which is consistent with the self-maximizing property. Then, we can see that authors with the same labels are highly related. For example, in the area of DM or IR, authors are highly relevant to each other. In addition, \ourmethod can also capture some underlying relevance between nodes. For instance, the author ``Il-Yeol Song'' is labeled with DB. However, he has also published lots of papers in the data mining area. Thus, he has relatively high relevance scores with authors in the DM area. From the aforementioned results, we can see that \ourmethod can not only accurately discover the relevance between nodes but is also able to capture the hidden relevance.

\begin{figure}[tbp]
    \centering
    \includegraphics[width=.8\columnwidth]{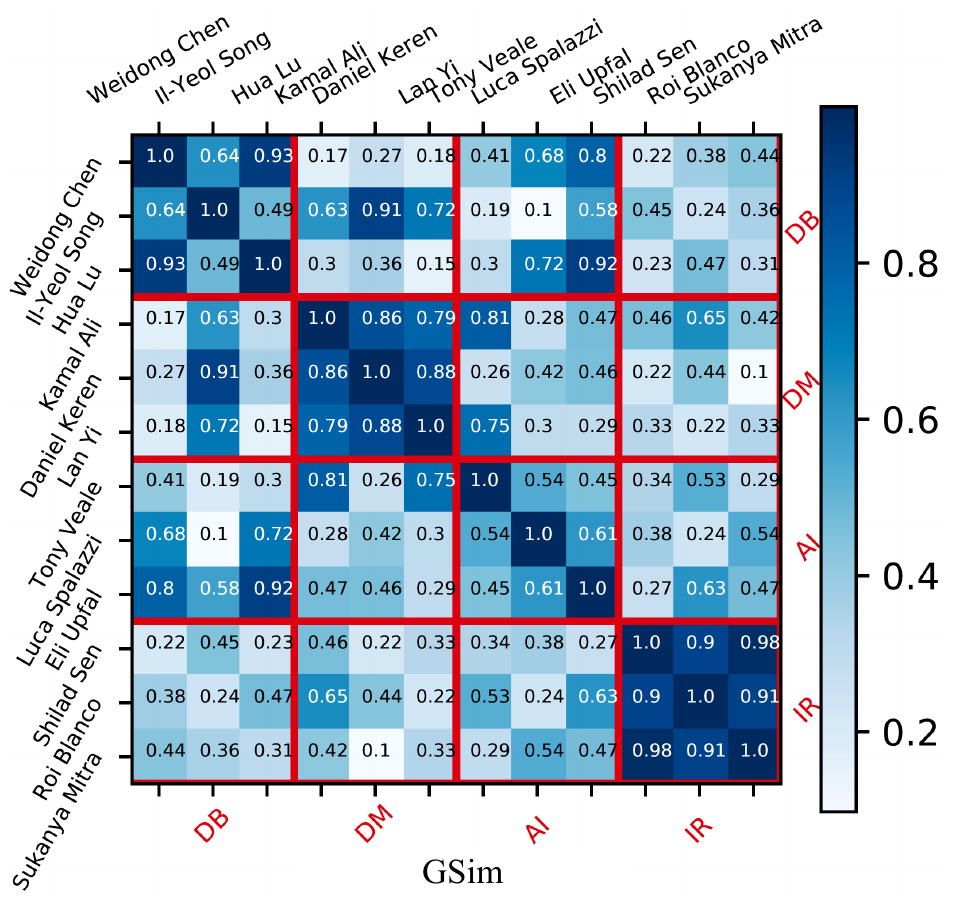}
    \caption{Relevance matrix between selected authors in DBLP-Multi dataset.}
    \label{fig:selected-matrix}
\end{figure}

\begin{figure}[tbp]
    \centering
    \includegraphics[width=.8\columnwidth]{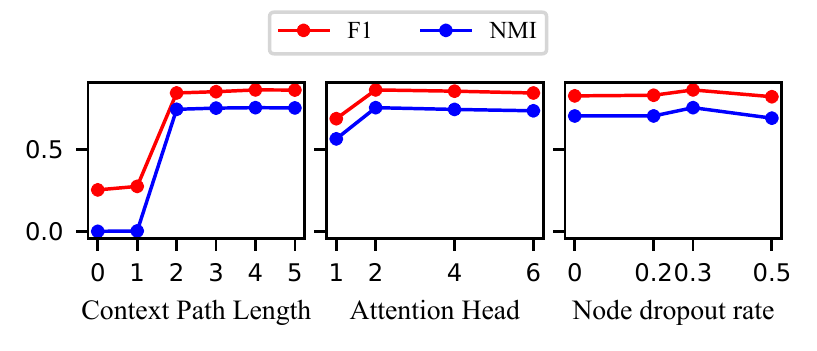}
    \caption{Parameter sensitivity w.r.t. different parameters.}
    \label{fig:parameter}
\end{figure}

\subsection{Parameters Analysis and Ablation Study}\label{sec:parameters}
In this section, we discuss the parameters analysis and ablation study.
As shown in Fig. \ref{fig:parameter}, with the increase of parameter values, \ourmethod's performances raise first and then drop slightly; the best performances are reached when the Context Path Length, Attention head, and Node dropout rate reach 4, 8, and 0.3, respectively. The reason is that an over large context path length may contain redundant semantics which deteriorates the performance. Although more attention heads can capture more diverse relation importance and increase the representation ability, they also introduce more parameters to the model, making it hard to train. Besides, too many nodes are dropped, which causes graph summary vectors cannot be well generated from the remaining nodes.

In Table \ref{tab:balanceloss}, we study the impact of balance between $\mathcal{L}_S$ and $\mathcal{L}_U$ in Equation \ref{eq:loss}. We can see that the performance of \ourmethod is the best when the balance is 1:1. This is because the two losses are complementary to each other and need to be equally considered during the training. When the balance is 0:1, the model only focuses on self-relevnace, which largely impairs performance. When the balance is 1:0, the model ignores the self-relevance, which also deteriorates the performance. Besides, reducing the weight ratio of both $\mathcal{L}_S$ and $\mathcal{L}_U$ results in a decrease in performance, emphasizing their equal importance in achieving optimal results.

In Table \ref{tab:ablation}, to evaluate the effectiveness of different components, we gradually remove the \textit{type-length attention}, \textit{relation attention}, and \textit{relation message passing}. From the results, we can see that the performance of the model decreases with components removed. Specifically, without the type-length attention and relation attention, \ourmethod fails to differentiate the importance of different relations. Moreover, without the relation message passing, \ourmethod cannot capture the asymmetric semantics as analyzed in section \ref{sec:cp-gnn+}.

\begin{table}[tbp]
    \caption{Balance between $\mathcal{L}_S$ and $\mathcal{L}_U$.}
    \label{tab:balanceloss}
    \center
    \resizebox{.65\columnwidth}{!}{
    \begin{tabular}{@{}ccccc@{}}
        \toprule
        $\mathcal{L}_S:\mathcal{L}_U$ & F1              & NMI             & ARI             & Purity          \\ \midrule
        1:1                           & \textbf{0.8856} & \textbf{0.7914} & \textbf{0.8539} & \textbf{0.9404} \\
        0:1                           & 0.2731          & 0.0285          & 0.0252          & 0.3736          \\
        1:0                           & 0.8801          & 0.7808          & 0.8386          & 0.9345
        \\
        0.8:1                         & 0.8827          & 0.7857          & 0.8421          & 0.9354          \\
        0.6:1                         & 0.8853          & 0.7911          & 0.8476          & 0.9379          \\
        1:0.8           &0.8807	&0.7821	&0.8393	&0.9349\\
        1:0.6           & 0.8803     & 0.7810 & 0.8387 &0.9340\\
        \bottomrule
    \end{tabular}%
    }
\end{table}

\begin{table}[tbp]
    \caption{Effectiveness of different components.}
    \label{tab:ablation}
    \center
    \resizebox{.65\columnwidth}{!}{
    \begin{tabular}{@{}ccc@{}}
        \toprule
        Method                                & NMI             & ARI             \\ \midrule
        GSim                                  & \textbf{0.8620} & \textbf{0.7539} \\
        \textit{w/o} type-length attention    & 0.8574          & 0.7497          \\
        \textit{w/o} relation attention       & 0.8405          & 0.7318          \\
        \textit{w/o} relation message passing & 0.8284          & 0.7104          \\ \bottomrule
    \end{tabular}%
    }
\end{table}

\begin{table}[tbp]
    \caption{Improvement of memory usage and training time brought by relation message passing.}
    \label{tab:memory}
    \center

    \resizebox{.9\columnwidth}{!}{
        \begin{tabular}{@{}ccc@{}}
            \toprule
            Method                                & Memory (Mb)          & Time/Epoch (s)       \\ \midrule
            \textit{w/o} relation message passing & 7090                 & 0.33                 \\
            GSim                                  & 2245                 & 0.12                 \\\midrule
            Improvement                           & $\downarrow$ 68.33\% & $\downarrow$ 63.63\% \\ \bottomrule
        \end{tabular}
    }
\end{table}
\subsection{Impact of Relation Message Passing}
In this section, we further study the impact of relation message passing. Relation message passing is proposed to measure nodes of heterogeneous types without increasing complexity. We show the improvement of GPU memory usage and training time brought by relation message passing in Table \ref{tab:memory}. From Table \ref{tab:memory}, we can see that, by using relation message passing, memory usage and training time are reduced by 68.33\% and 63.63\%, respectively. This is because the model $w/o$ relation message passing needs $2K$ layers to achieve the same results after adding intermediate nodes, which introduces additional memory consumption and training time.


\subsection{Relation Attention Visualization}




To further analyze whether \ourmethod can differentiate the context paths, we first present the corresponding relations attention matrix acquired from \ourmethod in Fig. \ref{fig:ACMMatrix} and \ref{fig:DBLPMatrix}, where each entry is the attention score of the relation with a source node type and a target node type. The attention score of each context path can be computed by summarizing the scores of its relations. Then, to justify whether the paths with higher attention scores are more meaningful for community detection, we adopt the Metapath2vec to evaluate the effect of each path. The results are shown in Fig. \ref{fig:ACMResult} and \ref{fig:DBLPResult}

For example, the paths PAP and PSP in ACM are both 2-length context path. Therefore, there attention scores can be computed from the relation attention matrix of 2-length context path shown in Fig. \ref{fig:ACMMatrix} where $S(PAP)=PA+AP=7$ and $S(PSP)=PS+SP=9$. Clearly, we can find that the attention score of PSP is higher than PAP, which means the path PSP is slightly more important than PAP for community detection. This can be justified by the result shown in Fig. \ref{fig:ACMResult} where the F1 score of PSP is higher than PAP.
Similarly, from Fig. \ref{fig:DBLPMatrix}, the attention scores of paths APCPA and APTPA are $S(APCPA)=AP+PC+CP+PA=23$ and $S(APTPA)=AP+PT+TP+PA=21$. This indicates that the relationship reflected by APCPA is a little bit more important than that of APTPA. This finding can be proved by the result shown in Fig. \ref{fig:DBLPResult} where the F1 of APCPA is higher than APTAP.

In summary, the above analysis demonstrates that relation attention can discover many context paths of different importance and capture the context path of higher importance that are more meaningful and useful for relevance measure.

\begin{figure}[]
    \centering
    \begin{subfigure}[b]{0.54\columnwidth}
        \includegraphics[width=1.\linewidth]{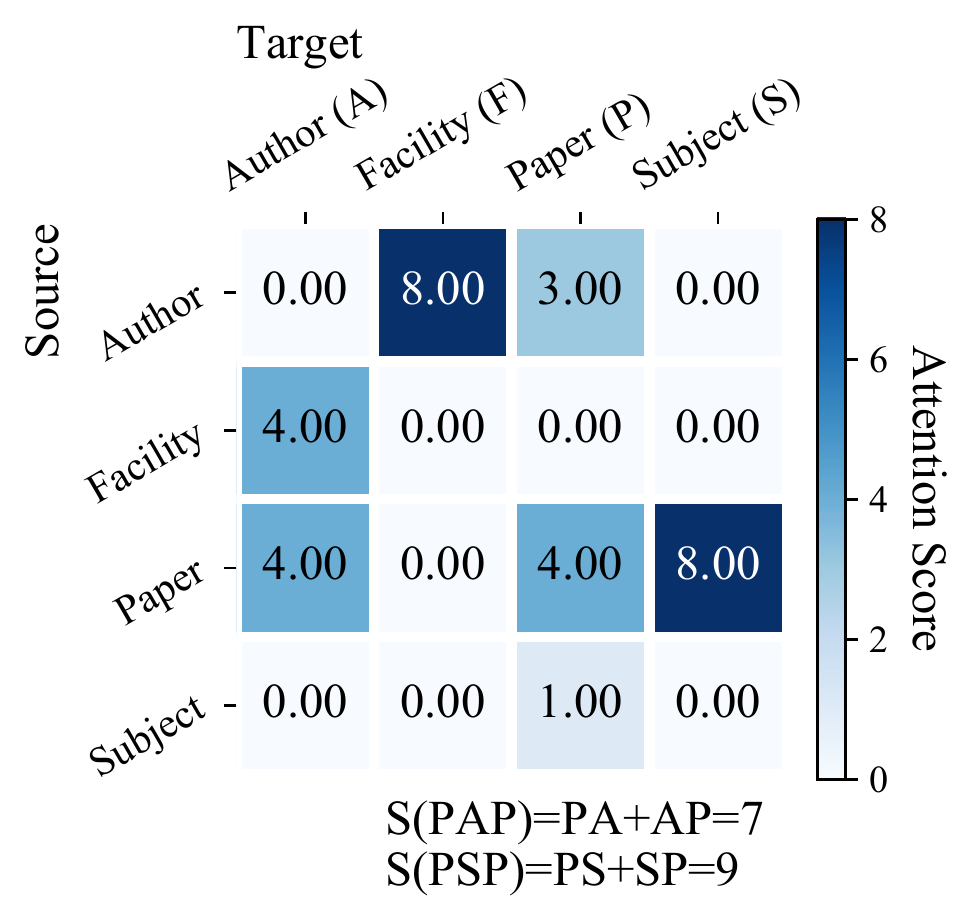}
        \caption{Attention matrix of 1-length context path on ACM.}
        \label{fig:ACMMatrix}
    \end{subfigure}
    \begin{subfigure}[b]{0.45\columnwidth}
        \includegraphics[width=1.\linewidth]{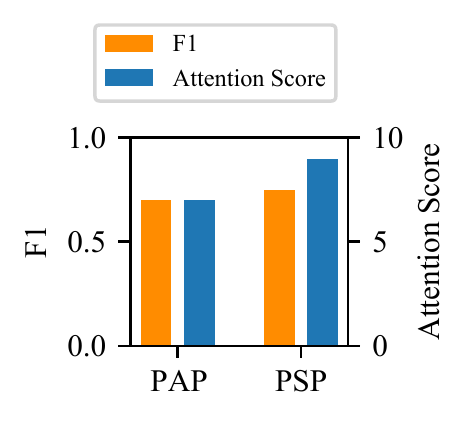}
        \caption{Metapath2vec F1 values on ACM.}
        \label{fig:ACMResult}
    \end{subfigure}
    \caption{Visualization of the context relation attention matrix on ACM.}
    \label{fig:rel_atten_matrix}
\end{figure}

\begin{figure}[]
    \centering
    \begin{subfigure}[b]{0.54\columnwidth}
        \includegraphics[width=1.\linewidth]{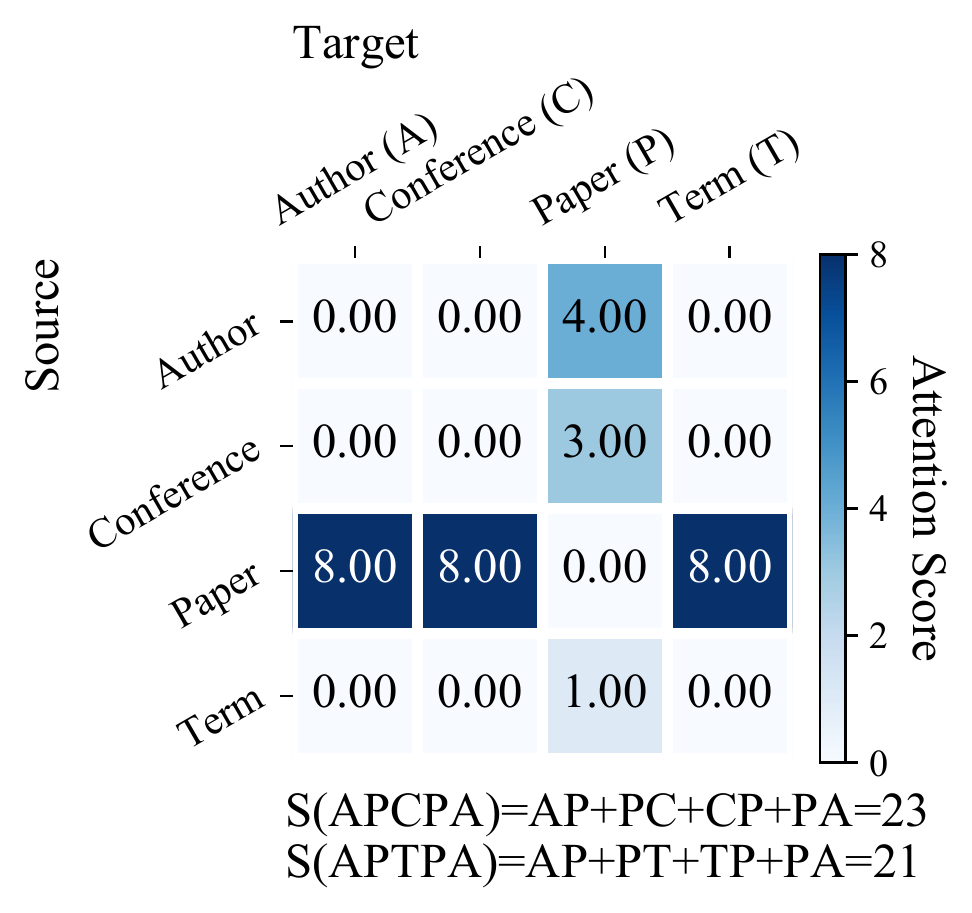}
        \caption{Attention matrix of 3-length context path on DBLP.}
        \label{fig:DBLPMatrix}
    \end{subfigure}
    \begin{subfigure}[b]{0.45\columnwidth}
        \includegraphics[width=1.\linewidth]{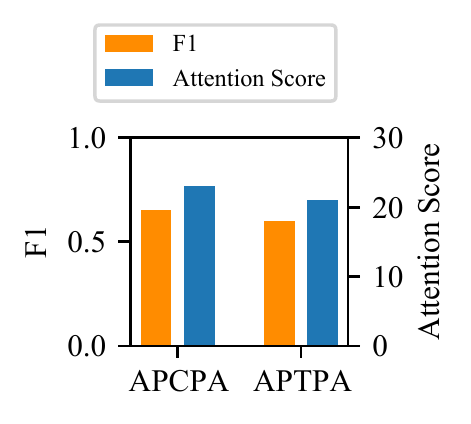}
        \caption{Metapath2vec F1 values on DBLP.}
        \label{fig:DBLPResult}
    \end{subfigure}
    \caption{Visualization of the context relation attention matrix on DBLP.}
    \label{fig:DBLP_rel_atten_matrix}
\end{figure}



%% file: sections/conclusion.tex
\section{Conclusion}\label{sec:conclusion}
In this paper, we propose a context path-based graph neural network model for relevance measure in heterogeneous graphs, called \ourmethod. \ourmethod can automatically capture the semantics among nodes, and differentiate their importance. Furthermore, we introduce the relation message passing mechanism that enables \ourmethod to measure relevance between nodes with different types. The effectiveness of \ourmethod is guaranteed by both theoretical analysis and empirical study.
Extensive experiments on four real-world datasets show that \ourmethod outperforms other baselines and reaches the best performance in both the community detection and relevance search task.

In the future, we will consider the relevance measure in the unsupervised setting. Besides, we will try to explore the potential of \ourmethod for other graph mining tasks, such as node importance measure and node ranking. These will require us to consider the fine-grinned relative relevance in the heterogeneous graphs.